\documentclass[oneside,11pt,a4paper]{article}
\usepackage{amssymb, amsmath, amsthm,color}
\usepackage[authoryear, round]{natbib}

 \newcommand{\EE}{\mathcal{E}} 
 \newcommand{\PP}{\mathcal{P}} 
 \newcommand{\CE}{\textnormal{CE}} 

\newtheorem{lemma}{Lemma}
\newtheorem{proposition}{Proposition}
\newtheorem{corollary}{Corollary}
\newtheorem{remark}{Remark}
\newtheorem{example}{Example}

\usepackage{graphicx}

\usepackage[usenames,dvipsnames]{xcolor}
\usepackage{subcaption}
\usepackage{caption}
\usepackage{setspace}

\begin{document}

\author{Anne G. Balter~~~~~~~Nikolaus Schweizer}

\title{Robust Decisions for Heterogeneous Agents via Certainty Equivalents\thanks{We thank Antje Mahayni, Peter Schotman, Hans Schumacher and participants at the Netspar International Pension Workshop 2021 for very helpful comments and discussions. Anne G. Balter, Department of Econometrics and Operation Research, Tilburg University, Tilburg, The Netherlands and Netspar. a.g.balter@uvt.nl. Nikolaus Schweizer, Department of Econometrics and Operation Research, Tilburg University, Tilburg, The Netherlands. n.f.f.schweizer@uvt.nl }}
\date{June 2021}
\maketitle

\begin{abstract}
We study the problem of a planner who resolves risk-return trade-offs -- like financial investment decisions -- on behalf of a collective of agents with heterogeneous risk preferences. The planner's objective is a two-stage utility functional where an outer utility function is applied to the distribution of the agents' certainty equivalents from a given decision. Assuming lognormal risks and heterogeneous power utility preferences for the agents, we characterize optimal behavior in a setting where the planner can let each agent choose between different options from a fixed menu of possible decisions, leading to a grouping of the agents by risk preferences. These optimal decision menus are derived first for the case where the planner knows the distribution of preferences exactly and then for a case where he faces uncertainty about this distribution, only having access to upper and lower bounds on agents' relative risk aversion. Finally, we provide tight bounds on the welfare loss from offering a finite menu of choices rather than fully personalized decisions. 
\end{abstract}

\section{Introduction}

\paragraph{Overview.} In this paper, we study the problem of a planner who resolves a risk-return trade-off on behalf of a collective of agents with heterogeneous preferences. Classical examples come from portfolio choice where, e.g., the managers of a mutual fund or the designers of a pension system make decisions that simultaneously affect the investments of many individuals. Ideally, every agent would receive a tailor-made investment solution that is optimal given his individual preferences.  Preference heterogeneity among investors is indeed a well-documented fact, implying that one-size-fits-all solutions may lead to significant welfare losses.\footnote{See, e.g., \cite{dahlquist2018asset}, \cite{alserda2019individual} and \cite{calvet2021cross}.} Some agents are more willing to take risks than others. Yet, for various reasons, a full personalization of investment plans may not be optimal either. For instance, there may be economies of scale in offering only a limited number of investment products to agents, thus reducing transaction costs or the costs of having products approved by a regulator. Offering only a limited number of options may also simplify communication with individual investors, leading to a more efficient exchange of information and to more robust choices. For instance, there is evidence that a Swedish pension reform that gave agents a choice between hundreds of different investment products lead to choices that were presumably suboptimal for many agents, see e.g. \cite{cronqvist2004design}. Agents were simply not capable to align such a large choice set with their preferences due to a lack of financial literacy and insufficient resources for gathering the information that would be necessary for an informed choice. 

While collective investment is our main application throughout the paper, the type of problem is more universal. Think of the development and regulation of a vaccine that is needed to end an economic lockdown due to an infectious disease. Depending on their preferences and exposures, agents may have heterogeneous opinions about the optimal thoroughness of the approval process of such a vaccine. Agents who suffer strongly from the economic lockdown may be in favor of introducing the vaccine after a relatively short period of development and testing, while others may be in favor of a longer development period and a smaller probability of harmful side effects. 

Our analysis is based on a stylized model where outcomes are lognormally distributed and their return and riskiness are controlled by a decision parameter that can be interpreted like a fraction of wealth invested into risky assets. Agents are expected utility maximizers who differ in their levels of constant relative risk aversion, their so-called \textit{risk types}. In this setting, we first derive optimal decisions of a planner who knows the distribution of risk types across the population of agents. We begin with the case where a single decision has to be made for the whole population. Afterwards, we characterize optimal choice menus in a setting where a fixed number of possible decisions are offered to agents. Here,  we compare two settings. In the first one, the planner groups agents by risk types and then optimizes the decision within each group. In the second setting, the planner  offers a menu of decisions and lets agents pick their preferred option. It turns out that the two settings are equivalent at the optimum. With the optimal choice menu in the second setting, agents choose those groups that would have been assigned to them in the first. 

Next, we study the welfare loss from offering agents only a finite number of  choices instead of a tailor-made solution for every preference type. We derive tight bounds on the resulting welfare loss which depend only on the number of  choices and on bounds on the support of the distribution of risk types. Finally, we study a situation where the planner does not know the distribution of risk types exactly. We characterize robust optimal choice menus for a planner who only has access to bounds on the support of the distribution of risk types. Relying on a game-theoretic concept of adversarial robustness, we provide an explicit expression for the decision menus the planner should offer. 

\paragraph{Preference Aggregation.}
A key ingredient of our approach is the way in which we aggregate preferences, formulating the planner's objective based on the objectives of the individual agents. We propose a tractable and intuitive approach which interpolates between two classical extremes, the utilitarian approach and the Rawlsian (or Pareto) approach. 

In a nutshell, our planner evaluates the outcome of a decision based on the resulting distribution of individual certainty equivalents across the population. The planner applies a concave utility function to this distribution of certainty equivalents to compute the welfare that arises from different decisions. He thus exhibits inequality averse preferences that are analogous to the classical expected utility formulation of risk averse preferences  \citep{vNM}. The idea of applying utility theory to social choice rather than choice under risk is old, going back e.g. to \cite{vickrey1945measuring}. 
However, in a classical utilitarian approach, it would be more common for the planner to consider the population distribution of individual utilities rather than certainty equivalents. In fact, a key result in utilitarian welfare economics, Harsanyi's Utilitarian Theorem \citep[see][]{Hammond1992}, suggests that the planner's objective should be a linear functional of the individual utilities, i.e., a weighted sum of utilities. 

A classical problem of utilitarian social preferences is the utility monster of \cite{nozick1974anarchy}, an agent whose (marginal) utility for resources is so great that it dominates the planners preferences. A utilitarian social planner might just give all resources to the agent who claims to like them the most, disregarding fairness concerns.  With a view towards practical applications, this problem is exacerbated by the fact that individual utilities are only identified up to affine transformations, i.e., up to addition and multiplication with numbers that may be arbitrarily large. In contrast, certainty equivalents are identified. They can be elicited from agents by asking the right incentivized questions. By translating agents' individual utilities into certainty equivalents, the planner converts them into the same monetary units before adding them up.  
This avoids the problems of the utility monster and of adding up incomparable quantities with unidentified scale. In the absence of risk, the planner prefers to give equal amounts to all agents rather than favoring those with a stronger preference for money -- like a utilitarian would. By measuring inequality in terms of the distribution of certainty equivalents, the planner accounts for heterogeneity in agents' risk appetite while consciously ignoring heterogeneity in agents' taste for money. 

The utilitarian approach is, of course, not the only way of formulating social preferences. Under the competing Rawlsian view  the planner would focus on the preferences of the agent who benefits the least from his decisions \citep{Rawls}.\footnote{This is related to the Pareto approach of only considering decisions that make all agents better off. A drawback of the Pareto approach is that it does not give a complete ordering of all possible decisions. When a true compromise between different interests has to be reached, the Pareto criterion is silent.} In our setting of choice under risk, this approach corresponds to a dictatorship of the most risk averse agent in the population because that agent has the lowest certainty equivalent from \textit{any} given lottery. By varying the curvature of the planner's utility function we  interpolate between more Rawlsian and more utilitarian approaches. In particular, in the limit of an infinitely concave utility function -- infinite inequality aversion -- the planner's utility converges to the Rawlsian dictatorship of the most risk averse agent. 

For some parts of our analysis, we assume that the planner aggregates certainty equivalents using a \textit{logarithmic utility function}. In a portfolio choice setting, this assumption has a natural interpretation of optimizing the population average of the certainty equivalent growth rate. The logarithmic assumption leads to two further simplifications of our analysis: When implementing a decision for a subset of agents, the planner's objective is equivalent to treating the average agent in the subset as the representative agent, maximizing only his utility. Moreover, the planner's preferences become time-consistent, avoiding common problems in dynamic decision making outside the expected utility paradigm.

\paragraph{Interpretations and Applications.}
Throughout the paper, the main interpretation of our model is that of a financial planner who acts on behalf of a collective of agents with heterogeneous risk preferences. However, there is some flexibility both in interpreting the planner's preferences and in the potential practical applications. 

One alternative interpretation is in terms of preference uncertainty of a single agent who is planning for himself.  It may take an agent years of learning to understand his own risk preferences well. The distribution of risk types in our model can be interpreted as reflecting an agent's beliefs about his own risk preferences at a given point in time. By applying our preference functional, an agent can make decisions under risk while taking into account uncertainty about his own risk preferences. In this interpretation, our model can be viewed as an adaption of \cite{klibanoff2005smooth}'s smooth ambiguity model from uncertainty about the distribution of risk to preference uncertainty.\footnote{Both models share the same two-stage structure. The agent first applies an inner utility function to the distribution of risk, holding a realization of the uncertain parameter fixed. Afterwards, the agent applies an outer utility function and averages out the  uncertain parameter. However, in our model, the uncertain parameter is not related to the distribution of risk. Instead, it is the inner utility function itself which is uncertain. We apply many different inner utility functions while they apply only one. Consequently, in their model it does not make a difference whether the outer lottery is viewed as a lottery over certainty equivalents or expected utilities if the outer utility function is  adjusted suitably. In contrast, in our setting it matters  whether the outer lottery is taken over utilities or over certainty equivalents as we propose.}

In a related interpretation of our model, there is a  planner who acts on behalf of a single agent. Due to limits on the amount of information that can be communicated, the planner has only imperfect knowledge of the agents preferences. Using our theory, the planner can explicitly take this uncertainty into account. For example, recently, there has been increased interest in ``robo-advisors'' \citep[see e.g.][]{d2019promises,robosurvey}, machine learning tools that assist investors in their decisions. To be effective, these tools need to gradually  learn the investor's preferences. Our results may be used to manage the uncertainty in this learning process, providing, e.g., worst-case optimal menus of possible investment decisions given limited preference information. 

Finally, besides the financial applications, there are various other situations that can be formalized in a similar way, trading off risk against return when designing a public good. Problems like designing national defense or choosing security standards in public transport can be thought of as problems of trading risks against expenses. In these applications, all agents in the collective are exposed to exactly the same threats so that a grouping by risk type is usually not possible. For instance, all agents in a country get the same national defense. Our results for the implementation of a single decision do apply however. In contrast, in the design of medical treatments or vaccines, trading off effectiveness or availability against potential side effects, it may be conceivable to design different products for agents with different risk types -- analogously to different investment strategies in a financial setting. Moreover, since any medical product needs approval from the relevant authorities, there will  typically be a constraint on the number of products on offer. A welfare optimum may thus consist of a small menu of products, lying somewhere between a one-size-fits-all and a fully personalized solution.

\paragraph{Related Literature.} Our paper mainly contributes to two literatures, the literature on preference uncertainty and the literature at the intersection of quantitative finance and social choice theory. The latter literature is concerned with problems like the collective investment problem which is our baseline application. Many of the more advanced problems studied in this literature such as sharing rules \citep[e.g.][] {jensen2016suboptimal,brangeretal} or generation effects \citep[e.g.][]{Schumacher} are beyond the scope of this paper. Our main contribution to this literature is relatively foundational, rethinking the planner's objective and proposing to optimize the distribution of certainty equivalents rather than the utility of a representative agent or a weighted sum of utilities in the spirit of Harsanyi's Utilitarian Theorem.\footnote{See \cite{Chenetal} and the references therein for   recent applications of the utilitarian approach to collective investment, and \cite{Schumacher} for more discussion of Rawlsian vs. utilitarian objectives.} 

We are aware of only a few previous papers in quantitative finance where the investor's objective is based on the cross-section of certainty equivalents. In \cite{DesmettreSteffensen}, an investor optimizes a sum of certainty equivalents which is interpreted in terms of preference uncertainty rather than preference heterogeneity. The focus is on resolving the resulting time inconsistency problems, see the final part of Section \ref{secdynamic} for more discussion and \cite{kryger2010some} for earlier work in this direction.\footnote{\cite{DesmettreSteffensen} also provides further pointers to the earlier literature.} The cross-section of certainty equivalents also plays an important role in financial applications of the smooth ambiguity approach as in \cite{Balteretal}. There, however, a cross-section arises due to uncertainty about the correct financial market model rather than heterogeneity in preferences.  A second novelty of our approach within this literature is to analyze the impact of grouping investors by risk type. 

The other literature to which we contribute is the literature on model uncertainty and robustness, which has been very active in the past decades in various fields such as operations research, quantitative finance and in economics.\footnote{See \cite{ben2009robust}, \cite{FollmerSchied} and \cite{hansen2008robustness} for seminal monographs on the topic from these three respective fields.} Within this literature, a comparatively small subliterature applies robust optimization ideas to uncertainty about preferences. For example, \cite{armbruster2015decision} analyze the optimization of worst-case certainty equivalents when the utility function is only known in a few points. 
Our baseline analysis of a distribution of (constant) relative risk aversion parameters can be understood as an analogue of the smooth ambiguity approach applied to preference uncertainty. Our later results correspond to a worst-case analysis with a minimax regret criterion in the spirit of \cite{bell1982regret} and \cite{loomes1982regret}. Finally, some recent applications of preference uncertainty have appeared  in the context of robo-advising, e.g. in \cite{cap1} and \cite{cap2}, but, to our knowledge, none of these papers is closely related to ours in terms of the actual analysis -- implying that there is scope for future work bringing these literatures together.

\paragraph{Structure.} Section \ref{secsetting} introduces our baseline setting. Section \ref{secopt} characterizes optimal decisions for a planner who knows the distribution of risk types, first for a one-size-fits-all decision that is the same for all agents and then for menus of decisions that are tailored to groups of agents. In Section \ref{secbounds}, we provide robust bounds on the welfare loss from implementing a finite menu of decisions rather than fully personalized solutions. Section \ref{secrobust} provides robust decision strategies for a planner who is uncertain about the distribution of risk types. Finally, Section \ref{secdynamic} shows how a dynamic multi-asset investment problem can be embedded into our  static baseline model. All proofs are in the appendix.

\section{The Setting}\label{secsetting}

In our model, a social planner faces a unit mass of agents who differ in their risk preferences. Agents are characterized by their risk type $\gamma \in \mathbb{R}_+$ which is distributed according to a distribution function $F$. Each agent faces a risky, non-negative payoff $R(m,Z)$ where $m\in \mathbb{R}$ is a decision implemented by the planner and the risk factor $Z$ is a random variable with commonly known distribution. The choice of $m$ should be thought of as a risk-return trade-off with higher values of $m$ implying  higher returns at higher risk. We are interested in situations where the planner can tailor $m(\gamma)$ to an agent's risk type to some extent. However, there is a constraint on the number of values the function $m(\gamma)$ may take, i.e., on the number of possible decisions the planner can offer to different agents. 

Throughout, we denote by $E[\cdot]$ the expected value with respect to the distribution of $Z$ and by $\EE[\cdot]$ the expected value with respect to the distribution $F$ of $\gamma$.\footnote{While $\EE[\cdot]$ is mathematically an expected value, its interpretation is more like a weighted sum over the agents in a population.} For the associated probabilities, we write $P(\cdot)$ and $\PP(\cdot)$ respectively. We assume that agents are risk averse expected utility maximizers. In particular, an agent with risk type $\gamma$ has a strictly increasing and strictly concave utility function $u_\gamma:\mathbb{R}^+\rightarrow \mathbb{R}$ and ranks payoffs according to their certainty equivalent 
\begin{equation}\label{objag}
\CE(\gamma, m)=u_{\gamma}^{-1}(E[u_{\gamma}(R(m,Z))]).
\end{equation}
We assume that the planner aims at optimizing the distribution of agents' certainty equivalents by choosing $m(\gamma)$ in a way that maximizes the functional
\begin{equation}\label{obj}
\EE[v(\CE(\gamma, m(\gamma))]).
\end{equation}
Here $v$ is  a strictly increasing function. When $v$ is linear, the planner optimizes the average certainty equivalent. Concavity of $v$ reflects an aversion against inequality among agents' certainty equivalents, while a planner with a convex $v$ is willing to sacrifice the certainty equivalents of some agents to the benefit of those with the highest certainty equivalent.  

We leave the planner's problem relatively general while making fairly concrete parametric assumptions on the distribution of payoffs and on agents' risk preferences. 
We assume that $R(m,Z)$ is of the form 
\begin{equation}\label{RmZ}
R(m,Z)=\exp\left(rT+(\mu-r)mT-\frac{1}{2}\sigma^2m^2 T +m\sigma Z \sqrt{T} \right)
\end{equation}
where $Z$ is standard normally distributed and $r$, $\mu$, $\sigma$ and $T$ are positive constants with $\mu > r$. These parametric assumptions can be motivated from a classical finance literature on optimal dynamic investment as discussed in detail in Section 
\ref{secdynamic}. In that interpretation, $R(m,Z)$ is the realized return after time $T$ for an agent who constantly reinvests a fraction $m$ of his wealth\footnote{While $m$ can be interpreted as a fraction of wealth that is invested into the risky asset for $m\in[0,1]$, we do not impose these constraints. In line with a large literature, we allow for short selling, $m<0$, and for buying stocks from borrowed money, $m>1$ in principle.} into a risky asset, which is a geometric Brownian motion with drift $\mu$ and volatility $\sigma$, while the remainder is invested into a riskless asset with interest rate $r$.\footnote{In our model formulation, we normalize the initial wealth of all risk types to 1. In Section \ref{secdynamic}, we argue that this assumption is without loss of generality.} 

With a single agent, the problem of choosing the optimal $m$, trading off higher risks against higher returns, is known as the Merton problem in finance. More generally, \eqref{RmZ} is a  tractable parametric formulation of risk-return considerations which can easily be interpreted outside the financial setting. Choosing a higher value of $m$ increases the return but also the risk that is inherent in the random payoff $R(m,Z)$. The next remark summarizes some properties of $R(m,Z)$.

\begin{remark}
We can split the  payoff $R(m,Z)$ into a deterministic factor $D(m)$ capturing returns and a stochastic factor $Y(m,Z)$ capturing risk, $R(m,Z)=D(m) Y(m,Z)$ where 
\[
D(m)=\exp\left(rT+(\mu-r)mT\right) = \exp\left( (1-m)rT+m \mu T\right)
\]
and 
\[
Y(m,Z)=\exp\left(-\frac{1}{2}\sigma^2m^2 T +m\sigma Z \sqrt{T} \right).
\]
By increasing $m$ from $0$ to $1$, the exponential growth rate in the term $D(m)$  increases from the riskless baseline $r$ to the higher rate $\mu$. With general $m>0$, $D(m)$ can reach exponential growth at any positive rate.  The price to pay for a higher rate is that risk, as captured by the term $Y(m,Z)$, increases with $m$. To see this, note first that the term $-\frac{1}{2}\sigma^2m^2 T$ in the exponent is chosen in such a way that $E[Y(m,Z)]=1$ for all $m$. In this sense, varying $m$ does not affect the scale of $Y$. It does however affect its riskiness as the variance of $Y(m,Z)$ increases with $m$,
\[
\textnormal{Var}(Y(m,Z))= \exp(m^2 \sigma^2 T). 
\] 
Thus, increasing $m$ increases returns $D(m)$, leaves $E[Y(m,Z)]$ unchanged but increases risk as captured by the variance of $Y(m,Z)$. 
\end{remark}

Regarding the distribution of risk preferences, we assume that agents with risk type $\gamma$ have a power utility of the form 
\[
u_\gamma(r)=\frac{r^{1-\gamma}-1}{1-\gamma} 
\]
for $\gamma\neq 1$ and, as usual, $u_\gamma(r)=\log(r)$ for $\gamma =1$. We assume that the distribution function $F$ of $\gamma$ is continuously differentiable with derivative $f$. The density function $f$ is assumed to be strictly positive over the support $[a,b]$ of $\gamma$ where $a>0$ and $b < \infty$. Agents' risk types thus correspond to constant relative risk aversions. They are assumed to be bounded away from the risk neutral and the infinitely risk averse cases $\gamma= 0$ and $\gamma = \infty$. While we leave the planner's preferences more general until further notice, we will occasionally assume a power utility here as well, 
\[
v(c)=\frac{c^{1-\eta}-1}{1-\eta} 
\]
for $\eta\neq 1$ and $v(c)=\log(c)$ for $\eta =1$ where $\eta$ is the planner's inequality aversion parameter. The next lemma collects some facts about the preferences of an agent with risk type $\gamma$.

\begin{lemma}\label{lem1}
The certainty equivalent of an agent with risk type $\gamma$ and implemented decision $m$ is given by
\begin{align}\label{CEgm}
\CE(\gamma, m)=\exp\left(r T + (\mu-r) m T - \frac12 \gamma m^2 \sigma^2 T\right).
\end{align}
The individually optimal decision for such an agent is given by 
\[
m^*(\gamma)=\frac{\mu-r}{\sigma^2\,\gamma}.
\]
\end{lemma}

The function $m^*$ corresponds to the famous investment fraction from the Merton problem. It consists of a return-risk ratio which is dampened by the individual risk aversion $\gamma$. We also define the inverse mapping $g^*(m)$ which maps a non-negative decision $m$ to the risk aversion level under which this decision is optimal,
\begin{equation}\label{gstar}
g^*(m) = \frac{\mu-r}{m \sigma^2},
\end{equation}
and thus $g^*(m^*(\gamma))=\gamma$. 

\begin{remark}
Inspecting equation \eqref{CEgm}, we see that for fixed $\gamma$ and $m$ the certainty equivalent exhibits an exponential growth behavior in $T$ at a rate given by the so-called certainty-equivalent growth rate 
$
\frac{1}{T} \log(\CE(\gamma,m)).
$
For the special case of a planner with a logarithmic utility function, $v(c)=\log(c)$, the planner's objective \eqref{obj} can be written as
$
\EE[ \log(\CE(\gamma,m(\gamma))) ].
$
Thus, in this case, the planner's objective is equivalent to maximizing the population average of the certainty-equivalent growth rate. 
\end{remark}

\begin{remark}
We have formulated the setting in such a way that there is a single random variable $Z$ which captures risk for all agents regardless of their risk type. The payoffs of all agents are perfectly correlated (up to deterministic transformations). This  assumption is without loss of generality. The planner evaluates joint distributions of risk types and random payoffs by their implied distributions of certainty equivalents, computing a certainty equivalent for each risk type before  aggregating. Thus, the results of the planner's calculation are identical for any dependence structure between agents' random payoffs. If all agents are indifferent between two payoff profiles, the planner is indifferent as well.

From an applied perspective, different dependence structures are plausible. When agents invest in the stock market and $m$ captures the riskiness of their strategy, assuming one common market risk factor for all agents is a simplifying but reasonable assumption. Yet when $m$ captures the dosage of a medical treatment and $Z$ captures potential side effects, side effects may well be independent across agents. Every agent then has their own independent copy of $Z$  which determines whether this agent suffers from side effects or not. Since the planner's preferences do not distinguish between dependent and independent risks across agents, we focus on the notationally simpler case of a single risk factor. 
\end{remark}

\section{Optimal Strategies}\label{secopt}

In this section we characterize optimal decision strategies for the planner. 
We begin with the case where the function $m(\gamma)$ can only take a single value, i.e., there is a single one-size-fits-all decision that is implemented for all agents. We then move on to the more flexible situation where a menu of $n$ possible decisions is offered to the agents. 

\subsection{One-size-fits-all Decisions}

Recall that $[a,b]$ denotes the support of $\gamma$ and that, due to monotonicity, all agents' preferred decisions lie in the interval $[m^*(b), m^*(a)] \subset (0,\infty)$. The following lemma characterizes the optimal decision if the same choice is implemented for all agents.

\begin{lemma}\label{OS1}
There exists a maximizer $m^*(a,b)\in [m^*(b), m^*(a)]$ of $\EE[v(\CE(\gamma, m)])$. The maximizer $m^*(a,b)$ is a solution to the equation 
\begin{equation}\label{mstar}
m^*(a,b) =   \frac{\mu -r}{\sigma^2\, \Gamma(a,b)}
\end{equation}
where
\[
\Gamma(a,b)=\EE\left[ \gamma\; \frac{h(\gamma, m^*(a,b))  }{\EE[h(\gamma, m^*(a,b))]}\right]\in [a,b]\;\;\;
\text{ and }\;\;\;h(\gamma,m)=\CE(\gamma,m)v'(\CE(\gamma,m)).
\]
\end{lemma}

The definition of the optimal decisions $m^*(a,b)$ in the lemma is implicit: $m^*(a,b)$ is a Merton fraction for some level of risk aversion $\Gamma(a,b)$ in the support $[a,b]$ of $\gamma$. $\Gamma(a,b)$ can be interpreted as the expected value of $\gamma$ under some change of measure proportional to $h$. However the change of measure itself depends on $m^*(a,b)$. With an implicit definition like this, existence and uniqueness of solutions are not clear a priori. The lemma shows existence of an optimal strategy which solves the first order condition \eqref{mstar}. Yet, one can construct examples in which \eqref{mstar} has multiple solutions, corresponding, e.g., to local minima or maxima. When there are multiple \textit{global} maxima, we assume throughout that  $m^*(a,b)$ is the smallest maximizer. The next lemma treats the case of power utility functions.
\begin{lemma}\label{OS2}
Suppose that $v$ is a power utility function with parameter $\eta$. Then optimal decisions  $m^*(a,b)$ are characterized as solutions to the equation 
\begin{equation}\label{mstareta}
m^*(a,b) =   \frac{\mu -r}{ \sigma^2\,\Gamma(a,b)} 
\end{equation}
where
\begin{equation}\label{exptilt}
\Gamma(a,b)=\EE\left[ \gamma\; \frac{ \exp(\gamma \theta(m^*(a,b)))}{\EE[\exp(\gamma \theta(m^*(a,b)))]}\right]\in [a,b]
\end{equation}
and $\theta(m) = \frac{1}{2}\sigma^2 (\eta-1)T m^2$. Moreover,
\item[(i)] in the logarithmic case, $\eta=1$, $\Gamma(a,b)= \EE[\gamma]$ and thus the optimal strategy is given by $m^*(a,b)= m^*(\EE[\gamma] )$. 
\item[(ii)] In the case $\eta>1$, there exists a unique solution $m^*(a,b)$ to   \eqref{mstareta}. The solution satisfies $m^*(a,b) < m^*(\EE[\gamma] )$.
\item[(iii)] In the case $\eta \in (0,1)$, any solution to  \eqref{mstareta} is greater than $m^*(\EE[\gamma])$. In particular, an optimal decision $m^*(a,b)$ satisfies 
 $m^*(a,b) > m^*(\EE[\gamma] )$. 
\end{lemma}

A planner with logarithmic utility will thus implement the preferred solution of an agent with risk type $\EE[\gamma]$, the average risk type. The optimal decision of this planner coincides with the decision of a planner who ignores the dispersion in risk attitudes and simply optimizes the utility of a representative agent whose risk version corresponds to the population average $\EE[\gamma]$. A planner who is more inequality averse, $\eta >1$, will implement a more risk averse decision, following the preferences of some risk type $\Gamma(a,b)$ between $\EE[\gamma]$ and $b$. Finally, a less inequality averse planner, $\eta <1$, will follow the preference of some risk type $\Gamma(a,b)$ between $a$ and $\EE[\gamma]$. In this case, there may be multiple solutions to the first order condition \eqref{mstareta} but (at least) one of them will be a global maximum. 

\begin{remark}\label{remTilt}
Inspecting formula \eqref{exptilt}, we see that $\Gamma(a,b)$ is the expected value of $\gamma$ under an alternative distribution that corresponds to an exponential tilting of the true distribution. Such exponentially tilted distributions naturally occur in the analysis of model uncertainty, see e.g. \cite{hansen2008robustness}, where they correspond to maximal and minimal expected values of $\gamma$ over a set of alternative models which lie within a relative entropy ball around the original model. The sign of the parameter $\theta$ determines whether a maximal or minimal expected value is computed. In our setting, $\theta$ is positive whenever $\eta$ is greater than 1. In this case, $\Gamma(a,b)$ is larger than $\mathcal{E}[\gamma]$, corresponding to a maximal expected value and  a distortion towards more risk averse types. The opposite happens for $\eta$ less than 1.
\end{remark}

\begin{remark}\label{remPC}
In Lemma \ref{OS2}, the exponential tilting constant $\theta$ depends on the length of the investment horizon $T$ except in the logarithmic case $\eta = 1$. In particular, the longer the time horizon, the stronger is the tilting. In the limit $T \downarrow 0$ of shorter and shorter time horizons, the optimal decision approaches the one from the logarithmic case. For $\eta >1$, this convergence will be from below. As the time horizon shortens, the decision becomes riskier. For $\eta <1$ the convergence is from above, corresponding to a gradual reduction in risk taking. The impact of inequality aversion is thus stronger on longer time horizons. 
\end{remark}

Except in the logarithmic case, optimal decisions  depend on the length of the time horizon $T$. Thus, for $\eta \neq 1$, the planner faces a time consistency problem when we move from static to dynamic decision making. This is discussed further in Section \ref{secdynamic}.

\subsection{Optimal Partitioning}

We now move to the case where the planner can implement a function $m(\gamma)$ which takes at most $n$ values, extending the case $n=1$ of the previous section. We compare two different versions of the planner's problem which we call the \textit{risk grouping} and the \textit{decision menu} setting. In the \textit{risk grouping} setting, the planner partitions  the support $[a,b]$ of $\gamma$ into $n$ subintervals. For agents from the same element of the partition, the same decision is implemented but decisions may vary from one partition element to the other. The planner  optimizes both the boundaries of the partition and the decision that is implemented within each partition element. In the \textit{decision menu} setting, a partition arises endogenously through agents' choices. The planner offers a menu of $n$ decisions and each agent picks his preferred option. As a main result, we show that the outcome of the optimal risk grouping solution is identical to the outcome of the optimal decision menu.

\paragraph{Risk Grouping.} For any $c<d$ with $[c,d] \subseteq [a,b]$, we define the optimal strategy $m^*(c,d)$ as in Lemma \ref{OS1} with the distribution of $\gamma$ replaced by its restriction to the subinterval $[c,d]$.\footnote{The density of this new distribution is thus equal to $f(g)/(F(d)-F(c))$ for $g\in [c,d]$ and 0 otherwise.} This is the optimal decision when attention is restricted to agents with risk types between $c$ and $d$. We consider partitions of $[a,b]$ given by boundaries $a=g_0<\ldots <g_n=b$. In the risk grouping setting, the planner can pick the numbers  $g_i$. In addition, he can pick numbers $m_1,\ldots,m_n$ such that $m(\gamma)=m_i$ for $\gamma \in [g_{i-1},g_i)$ and $i=1,\ldots n$. In line with our assumptions, his goal is to maximize
\begin{equation}\label{OBJn}
\EE[v(\CE(\gamma, m(\gamma)))] = \sum_{i=1}^n \int_{g_{i-1}}^{g_i} v(\CE(g, m_i))f(g) dg.
\end{equation}
For given interval boundaries, summand $i$ only depends on $m_i$ but not on $m_j$, $j\neq i$. Each summand is maximized by picking $m_i=m^*(g_{i-1},g_i)$ following Lemma \ref{OS1}. 
This reduces the planner's problem to finding an optimal partition $(g_i)_i$. The next lemma characterizes optimal partitions, showing that they satisfy  a \textit{harmonic mean condition}. Recall that the harmonic mean between two positive real numbers $x$ and $y$ is given by 
\[
\mathcal{H}(x,y)= \frac{2}{\frac{1}{x}+\frac{1}{y}}
\]
and that $g^*$ from \eqref{gstar} maps a decision $m$ to the risk type  $g^*(m)$ who finds it optimal. 

\begin{lemma} \label{lemHM1}
Suppose the partition $g_0,\ldots,g_n$ with associated decisions $m_i=m^*(g_{i-1},g_i)$ is optimal in the sense of maximizing \eqref{OBJn}. Then we have for all $i=1,\ldots n-1$ \begin{equation}\label{HMcondition}
g_i = \mathcal{H}( g^*(m_i), g^*(m_{i+1}) ). 
\end{equation}
\end{lemma}

The harmonic mean condition \eqref{HMcondition} follows directly from the first order condition for optimal partitions, trading off the consequences of moving a marginal agent from one group to the other. In an optimal partition, a risk type who is at the boundary between two intervals must lie at the harmonic mean between the risk types whose respective individually optimal decisions are implemented in the two intervals. 

\begin{remark}Since risk types are inversely proportional to decisions, the harmonic mean condition \eqref{HMcondition} for risk types is equivalent to an arithmetic mean condition for optimal decisions: The individually optimal decision of a risk type at the boundary must be the arithmetic mean between the decisions implemented in the two groups,
\[
m^*(g_i) = \frac12 m^*(g_{i-1},g_i) + \frac12 m^*(g_{i},g_{i+1}).
\]
\end{remark}

The following example of uniformly distributed risk types is visualized in Figure \ref{FIG:example}.

\begin{example}\label{exuni}
Suppose the planner has logarithmic utility and $\gamma$ is uniformly distributed on $[a,b]$. By Lemma \ref{OS2}, it follows that 
$$g^*(m^*(g_{i-1},g_i))=\Gamma(g_{i-1},g_i)=\frac{g_{i-1}+g_i}{2}.$$ 
Plugging this into the harmonic mean condition \eqref{HMcondition} and rearranging gives the relation $g_i=\mathcal{G}(g_{i-1},g_{i+1})$ where $\mathcal{G}(x,y)=\sqrt{xy}$ denotes the geometric mean. Since $g_0=a$ and $g_n=b$ it follows that the optimal partition is geometric, i.e.,  
$
g_i = a^{1-\frac{i}{n}} b^{\frac{i}{n}}.
$
Optimal decisions are thus given by
\[
 m^*(g_{i-1},g_i) = \frac{\mu-r}{\sigma^2 \Gamma(g_{i-1},g_i)} = \frac{\mu-r}{\sigma^2 \left(\frac{g_{i-1}+g_i}{2}\right)} .
\]
\end{example}

\begin{figure}[!h]
	\begin{subfigure}{.49\textwidth}
		\centering
		\includegraphics[width=0.9\linewidth]{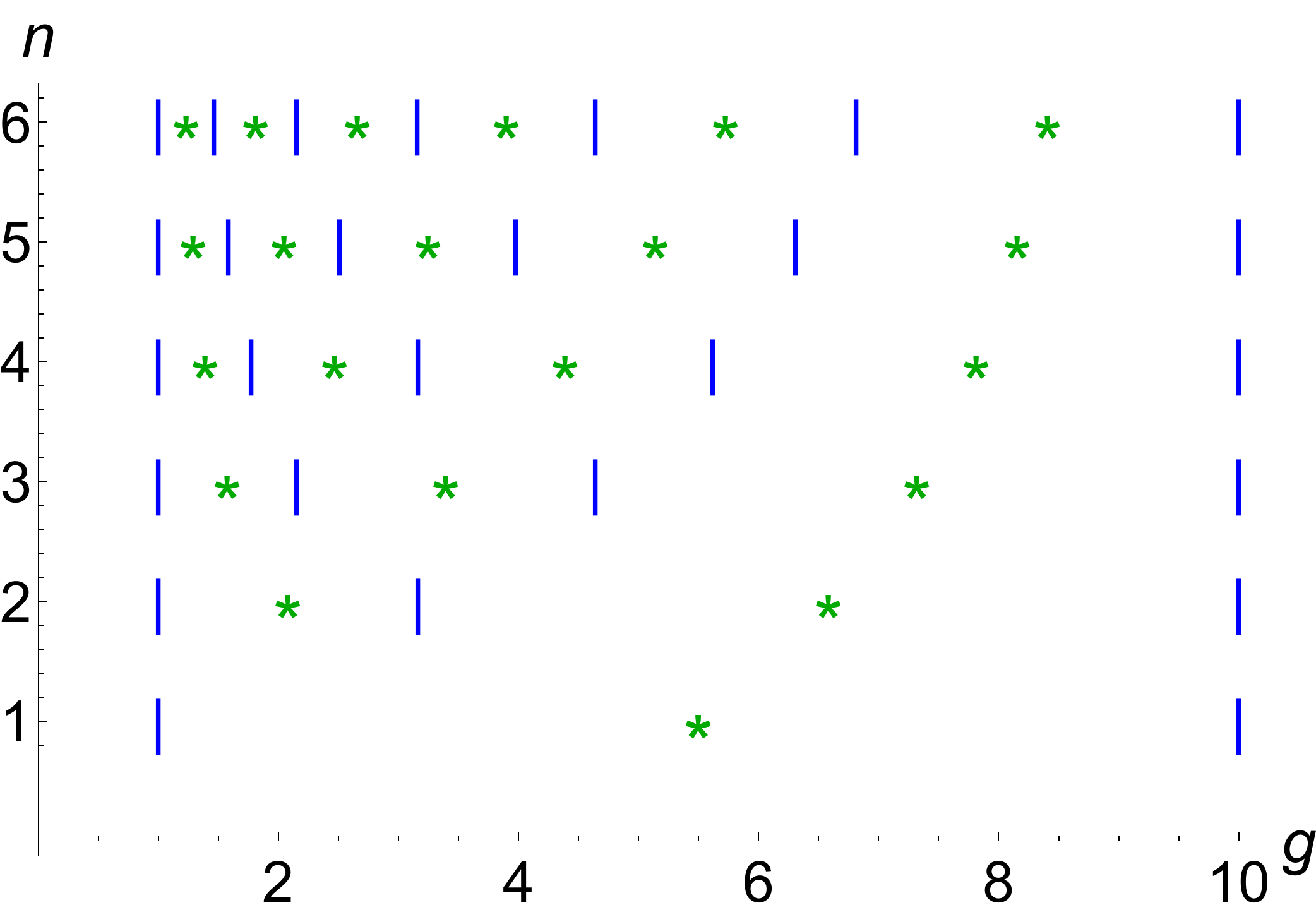}  
		\caption{$\Gamma_i$ (green stars) and $g_i$ (blue lines).}	
		\label{FIG:g}
	\end{subfigure}
	\begin{subfigure}{0.49\textwidth}
		\centering
		\includegraphics[width=0.9\linewidth]{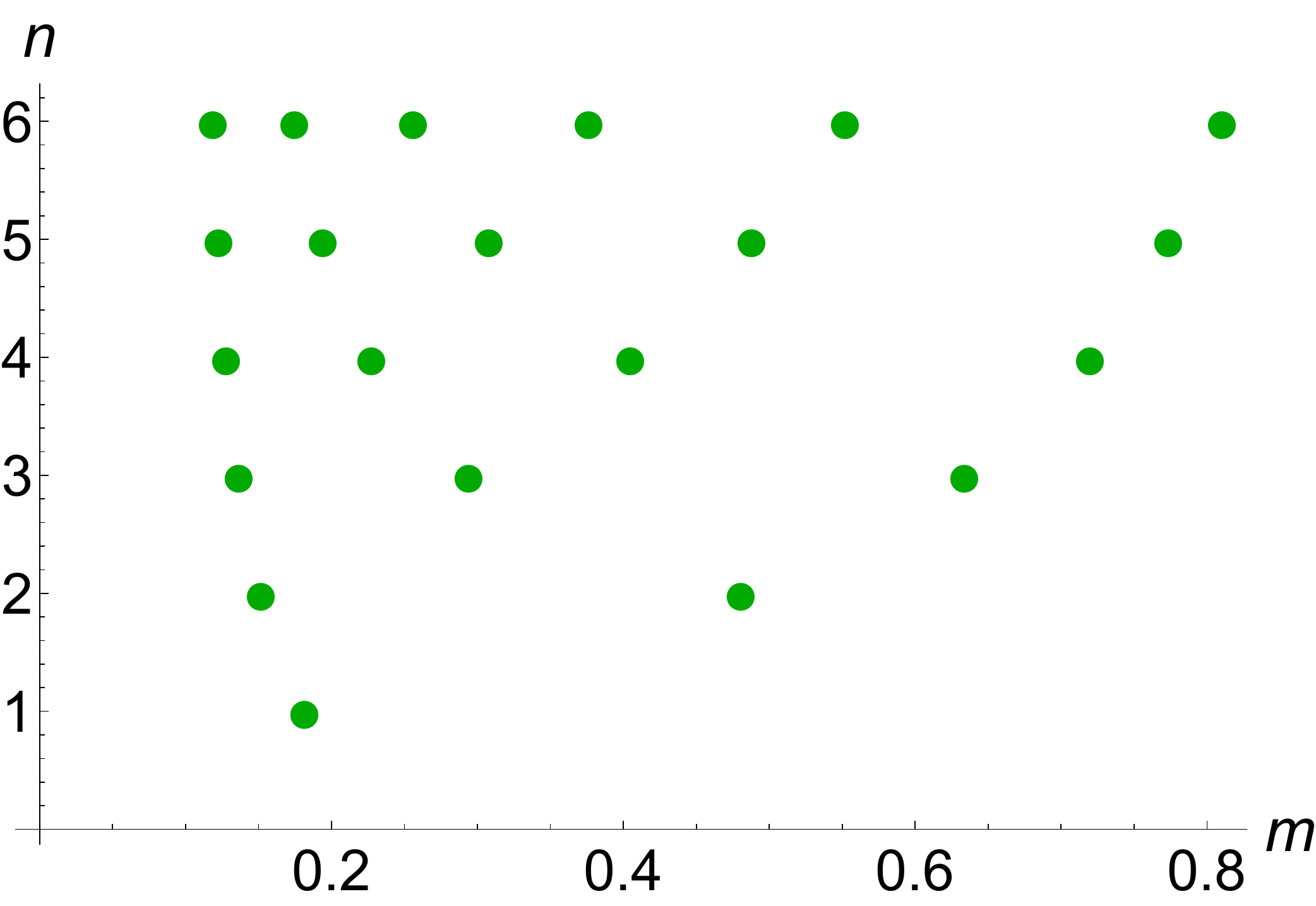}  
		\caption{$m^*(g_{i-1},g_i)$ (green dots).}	
		\label{FIG:m}
	\end{subfigure}
	\caption*{ \footnotesize  For varying $n$, the left panel shows the targeted risk types $\Gamma_i=\Gamma(g_{i-1},g_i)$ and the resulting partition boundaries $g_i$  as given in Example \ref{exuni} for the case of uniformly distributed risk types. The right panel shows the corresponding choice menus $ m^*(g_{i-1},g_i)$. The parameters are $a=1$, $b=10$ and $(\mu-r)/\sigma^2=1$.}
	\caption{Optimal decisions for a uniform distribution of risk types.}
	\label{FIG:example}
\end{figure}

\paragraph{Decision Menus.} In the risk grouping setting,  the planner can assign agents to groups and then force a decision on each group. In many practical applications, a planner's power is more limited. For instance, it may be the case that the planner simply releases a menu of $n$ products which correspond to choosing decisions $m_1,\ldots, m_n$. Agents can pick any product they like from this menu. This is the \textit{decision menu} setting. 

Consider an agent with type $g \in [a,b]$ facing a menu of decisions $m_1 >\ldots > m_n$. Which one should he pick? To this end, the agent  needs to check where his preferred decision $m^*(g)$ stands compared to the $m_i$. If $m^*(g)\geq m_1$, the agent chooses $m_1$ and if $m^*(g)\leq m_n$, he chooses $m_n$. If $m^*(g)$ lies between $m_i$ and $m_{i+1}$, the agent chooses either of these, depending on an indifference calculation which is found in the next lemma.\footnote{In this discussion, we are agnostic about the behavior of agents who are indifferent between two decisions. As risk types are continuously distributed, the set of such agents has mass zero.} 

\begin{lemma} \label{lemHM2}
An agent with type $g_i \in [a,b]$ is indifferent between decisions $m_i$ and $m_{i+1}$
iff the harmonic mean condition
$
g_i = \mathcal{H}(g^*(m_i), g^*(m_{i+1}) )
$
is satisfied. 
\end{lemma}

Given a menu of possible decisions, agents will sort into a partition by choosing one of the two decisions that are closest to their preferred one. Lemma \ref{lemHM2} shows that the partitions that arise endogenously in this way satisfy the harmonic mean condition \eqref{HMcondition}. Thus, instead of prescribing a partition together with associated decisions as in Lemma \ref{lemHM1}, the planner can simply prescribe the corresponding decision menu. Agents then sort into the associated optimal partition by evaluating the harmonic mean condition.  

Thus, while not all pairs of partitions and decision menus will be aligned with agents' preferences, optimal partitions have this property. In the language of mechanism design, optimal decision menus are incentive compatible: No agent has an incentive to misreport their type to the planner to be assigned to a different group. The trade-offs the planner faces when designing the partition are aligned with the trade-offs the agents face themself when picking a group. This works despite the fact that agents care only about their own risk type and not about the entire distribution like the planner. 

\section{Bounding the Welfare Loss from Grouping}\label{secbounds}

In our model, the planner is restricted in the number $n$ of possible decisions he can offer to agents. However, in order to maximize welfare, it would be optimal to offer to each agent the individually optimal decision $m^*(\gamma)$. In this section, we study the welfare loss from being forced to using a finite menu of choices. In particular, we derive sharp bounds which show how this loss depends on the number of groups $n$ and the relative difference between the extremal risk types, $b/a$. Throughout this section, we assume that the planner's utility function is logarithmic, $v(c)=\log(c)$. 

For the problem of this section, it is useful to think of strategies $m(\gamma)$ in terms of their associated implied risk aversion function $G(\gamma)$. This function maps an agent's risk type $\gamma$ to the risk type of an agent who prefers the strategy that $\gamma$ receives over all others,
\[
G(\gamma)= g^*(m(\gamma))=\frac{\mu-r}{\sigma^2 m(\gamma)}.
\]
We denote by $m_{n}^*(\gamma)$ the optimal strategy when the planner can offer a menu of $n$ different decisions as discussed in the previous section. Using that the planner has logarithmic utility, we know that for $n=1$
\[
m_{1}^*(\gamma) = \frac{\mu-r}{\sigma^2 \EE[\gamma]}
\]
so that the implied risk aversion function is constant, $G_1^*(\gamma)= \EE[\gamma]$. In the limiting case $n=\infty$ each risk type receives his individually optimal strategy so 
$G_\infty^*(\gamma)= \gamma$. This limiting case is our benchmark. For intermediate values of $n$, we know that optimal strategies are characterized by a partition $a=g_0<\ldots < g_n=b$ and by the fact that 
\[
G^*_n(\gamma)= \frac{\int_{g_{i-1}}^{g_{i}} g f(g) dg}{\int_{g_{i-1}}^{g_{i}} f(g) dg} = \EE\left[\gamma\left| \gamma \in [g_{i-1}, g_{i}] \right.\right]
\]
for $\gamma \in [g_{i-1}, g_{i}]$. Each agent's implied risk aversion is the mean risk aversion inside his partition element.  The next lemma rewrites the planner's objective in terms of $G$. 
\begin{lemma}\label{lemCEgrowth}
For any strategy $m(\gamma)$ with associated implied risk aversion function $G(\gamma)$, we can write the planner's utility as
\begin{equation}\label{eqCEg}
\frac{1}{T} \EE[\log(\CE(\gamma,m(\gamma)))]= r + \frac12\left(\frac{\mu-r}{\sigma}\right)^2 \, E
\end{equation}
where 
\[
E = \EE\left[ \frac{2}{G(\gamma)} -\frac{\gamma}{G(\gamma)^2} \right].
\]
\end{lemma}

\begin{remark}\label{eqeq}
Due to the planner's logarithmic utility, the left hand side in \eqref{eqCEg} corresponds to what is often called the ``certainty equivalent growth rate''. In our setting, since the curvature in $v$ reflects an aversion to inequality, the term ``equality equivalent growth rate'' would be more appropriate.
\end{remark}

From the lemma we see that the planner's logarithmic utility grows linearly with $T$ at a rate that consists of the interest rate $r$ plus an additional term. This term consists of two factors, the square of the Sharpe ratio $(\mu-r)/\sigma$ which captures properties of the market environment and a second factor $E$ which depends on the function $G$. 

This factor $E$ is the topic of the remainder of this section. It is the preference-dependent part of the planner's growth rate. The  next lemma derives an expression for 
\begin{equation}\label{ENformula}
E_n^* = \EE\left[ \frac{2}{G_n^*(\gamma)} -\frac{\gamma}{G_n^*(\gamma)^2} \right]
\end{equation}
which is the value of $E$ for the optimal strategies associated with different values of $n$. 
\begin{lemma}\label{lemRepr}
We can write 
\begin{equation}
E_n^*\; = \sup_{a=g_0<\ldots < g_n=b
} \;
\sum_{i=1}^n \frac{\PP\left(\gamma \in [g_{i-1},g_{i}] \right)}{\EE\left[\gamma \left| \gamma \in [g_{i-1},g_{i}]\right.\right]}.\label{eqRepr}
\end{equation}
\end{lemma}

Here, the supremum runs over all admissible $n$-element partitions. In the boundary cases of $n=1$ and $n=\infty$, formula \eqref{ENformula} implies even simpler expressions for $E_n^*$ as there is no dependence on an unknown optimal partition,   
\begin{equation}
E_1^* = \frac{1}{\EE[\gamma]}\;\;\;\text{ and } \;\;\;E_\infty^* = \EE\left[\frac{1}{\gamma}\right].\label{EE1inf}
\end{equation}
Jensen's inequality implies that, as expected,  $E_1^*\leq E_\infty^*$ -- there is a welfare loss from having a one-size-fits-all decision rather than individualized optimal decisions. The next lemma provides an inequality in the opposite direction, thus quantifying this welfare loss. 
\begin{lemma}\label{bd1}
We have the inequality
\[
E_\infty^*  \leq \frac{\frac{b}a+\frac{a}b +2 }{4} \; E_1^*.
\]
\end{lemma}
Thus, the welfare loss can be bounded in terms of the range $[a,b]$ of $\gamma$.
If $\gamma$ is distributed between $1$ and $10$, we know that 
$E_\infty^*  \leq 3.025 \; E_1^*$, so the planner loses a factor 3 in $E$ by providing a one-size-fits-all solution rather than personalizing. 

\begin{remark}\label{rembd1}In light of \eqref{EE1inf}, the inequality in Lemma \ref{bd1} is a general fact about random variables with bounded support. The inequality  is sharp in the  boundary case of a discrete distribution where $\gamma$ takes values $a$ and $b$ with equal probability. 
\end{remark}

We next extend the bound of Lemma \ref{bd1} from $n=1$ to general $n$. 
\begin{proposition}\label{bd2}
We have the inequality
\[
E_\infty^*  \leq \frac{\left(\frac{b}a\right)^\frac1n+\left(\frac{a}b\right)^\frac1n +2 }{4} \; E_n^*.
\]
\end{proposition}

We thus see that when increasing $n$ the factor in front of $E_n^*$ decreases so the inequalities become sharper until, in the limit, the right hand side becomes $E_\infty^*$ just like the left hand side. The inequality is thus again sharp. When $\gamma$ is distributed between $1$ and $10$, the lemma tells us, e.g., that $E_\infty^*  \leq 1.37 \; E_2^*$ and $E_\infty^*  \leq 1.09 \; E_4^*$. The constant in the inequality thus approaches 1 already with a moderate number of groups. It follows that  $n$ should depend logarithmically on the ratio $b/a$ to keep the relative welfare loss bounded:

\begin{corollary}\label{corbd}
If 
\begin{equation}\label{lbnR}
n\geq \frac{\log(b/a)}{\log(4R-3 )}
\end{equation}
for some $R\geq 1$ then 
\[
E_\infty^*  \leq R \, E_n^*.
\]
\end{corollary}

The corollary is illustrated in Figure \ref{FIG:welfare}. For different values of the relative welfare loss $R$, we plot the lower bound on $n$ from \eqref{lbnR} as a function of the heterogeneity in risk preferences as measured by the ratio $b/a$. Clearly, as $b/a$ increases, the menu size has to increase to keep the welfare loss stable at the level $R$.\footnote{Unlike the actual menu size, the lower bound is not restricted to integer values. In particular, we see that the lower bound may well be smaller than $1$ for $b/a$ not too large, indicating that offering less than one menu choice would be sufficient to guarantee a relative welfare loss of at most $R$. The way to interpret this result is that even with $n=1$ the actual welfare loss compared to individually optimal decisions is less than $R$. In this sense, a loss of $R$ would correspond to a fictitious situation with $n<1$.}  

\begin{figure}[!h]
	\centering
		\includegraphics[width=.5\linewidth]{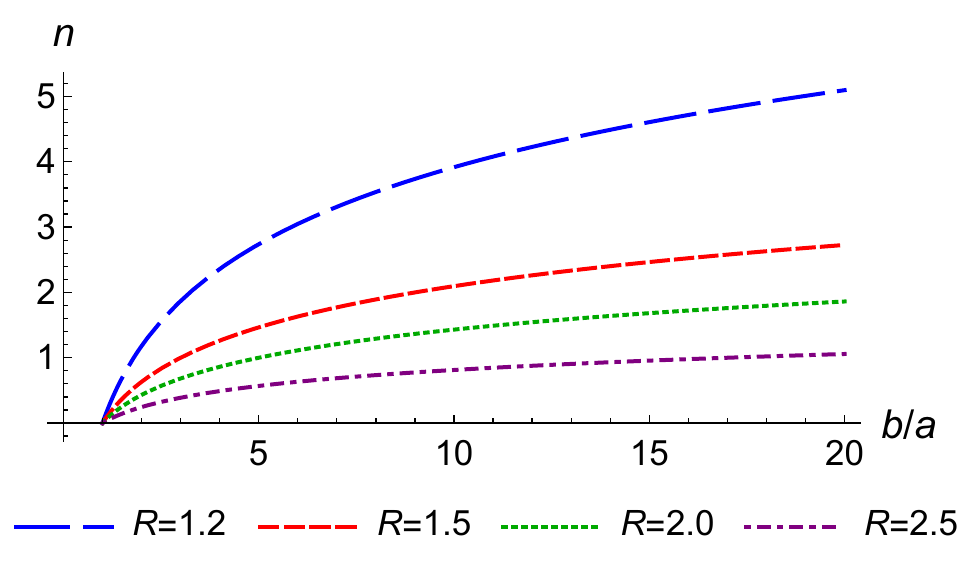}  
	\caption{Minimal menu size from \eqref{lbnR} for different values of $R$ as a function of $b/a$.}
	\label{FIG:welfare}
\end{figure}

\section{Robust Planning}\label{secrobust}

So far, we have assumed that the planner knows the distribution of risk types precisely. In this section, we relax this assumption and study optimal decisions of a planner who only knows that risk types lie in the interval $[a,b]$ but has no further knowledge about the distribution. We take an adversarial robustness approach, viewing the problem as a game between the planner and a fictitious adversary who chooses the distribution of risk types in a way that is least favorable to the planner. 

Throughout this section, we assume that the planner has logarithmic utility, $v(c)=\log(c)$. We begin again with the one-size-fits-all case where the planner chooses a single decision $m\in \mathbb{R}$ which applies for all risk types. Later, we also consider decision menus. We assume that the adversary chooses the distribution $F$ of $\gamma$ from the set $\mathcal{D}_{ab}$ of probability distributions with support in the interval $[a,b]$. Here, we do not restrict attention to continuous distributions with full support but also allow for atoms and for distributions which are concentrated in a single point. In particular, we denote by $F_x \in \mathcal{D}_{ab}$ the distribution which puts all  mass into $x\in [a,b]$.

It is easy to see that the result of Lemma \ref{OS2} carries over to this slightly more general setting: When $\gamma$ is known to be distributed according to $F\in \mathcal{D}_{ab}$, then it is optimal for the logarithmic planner to choose 
$$m^*_F=\frac{\mu-r}{\sigma^2 \EE_F[\gamma]}$$
where $\EE_F[\cdot]$ denotes the expected value over $\gamma \sim F$. 

\subsection{Robust One-Size-Fits-All Decisions}

We begin our game-theoretic analysis with a game we call the \textit{Absolute Criterion Game} (ACG). The ACG is a simultaneous-move zero-sum game in which the planner chooses $m \in \mathbb{R}$ with the goal of maximizing 
\[
\mathcal{A}(m,F)=\EE_F\left[v(\CE(\gamma, m)) \right]
\]
while the adversary chooses $F$ with the goal of minimizing $\mathcal{A}(m,F)$. The next lemma shows that the ACG has a somewhat trivial Nash equilibrium. 
\begin{lemma}\label{ACG}
In the unique Nash equilibrium of the ACG, the adversary chooses $F_b$ and the planner chooses $$m^*_{F_b}=\frac{\mu - r}{\sigma^2 b}.$$
\end{lemma}

The intuition behind the lemma is straightforward. For any fixed lottery, the certainty equivalent is minimal for the most risk averse agents. Thus, the adversary's best response to any strategy of the planner is to put all mass on the upper bound $b$, choosing $F_b$. In anticipation, the planner will act as if all agents had risk aversion level $b$. 

The adversarially robust decision strategy suggested by the ACG is not fully satisfying. Basically, the worst case generated by the adversary is just a situation in which the planner cannot achieve very much because agents are maximally risk averse. There is no remaining uncertainty.  The adversary does not try to fool the planner -- but instead gives him the chance to react optimally to the worst possible situation. In this way, the planner's decision targets a relatively extreme situation while underperforming everywhere else. These considerations motivate us to consider the \textit{Relative Criterion Game} (RCG). In the RCG, the planner maximizes the fraction of welfare that is attained compared to the welfare that could be attained if $F$ was known. The planner thus chooses $m$ to maximize
\[
\mathcal{R}(m,F)=\EE_F\left[v(\CE(\gamma, m)) \right] - \EE_F\left[v(\CE(\gamma, m^*_F)) \right]
\]
while the adversary chooses $F$ with the goal of minimizing $\mathcal{R}(m,F)$. The term that is different between the RCG and the ACG depends on $F$ but not on $m$. Thus, in moving from the ACG to the RCG, we have not changed the goals of the planner but only the ``success criterion'' of the adversary.

\begin{remark}
Due to the assumption that $v(c)=\log(c)$, we can rewrite $\mathcal{R}$ as a monotonic transformation of a ratio of ``equality equivalents''
\[
v^{-1}\left(\mathcal{R}(m,F)\right) =\frac{v^{-1}\left( \EE_F\left[v(\CE(\gamma, m)) \right]\right)}{ v^{-1}\left(\EE_F\left[v(\CE(\gamma, m^*_F))\right] \right)}.
\]
In this sense, $\mathcal{R}$ is a relative criterion. Moreover, in line with Remark \ref{eqeq}, $\mathcal{R}$ is the reduction in the ``equality equivalent growth rate'' due to uncertainty about $F$. 
\end{remark}

Inspecting the objective $\mathcal{R}(m,F)$, we see that it is non-positive, and that for fixed $F$, the planner can always achieve the optimal outcome of zero by implementing the strategy $m^*_F$, $\mathcal{R}(m^*_F,F)=0$. It follows that in any equilibrium the adversary must play a mixed strategy: It cannot be optimal for him to just implement a fixed $F$ because then the planner can react optimally with $m^*_F$. Instead the adversary has to randomize between different distributions of risk types. This is reflected in the unique Nash equilibrium of the RCG which is characterized in Proposition \ref{RCG}. 
\begin{proposition}\label{RCG}
In the unique Nash equilibrium of the RCG, the adversary plays a mixed strategy, choosing $F_a$ with probability 
\[
p^*=\frac{\sqrt{b}}{\sqrt{a}+\sqrt{b}}
\]
and $F_b$ otherwise. The planner plays the pure strategy
\[
m^*(\sqrt{ab})=\frac{\mu-r}{\sigma^2 \sqrt{ab}}.
\]
The resulting equilibrium value of $\mathcal{R}$ is given by  
\[
\mathcal{R}^*=- \,\frac{(\mu-r)^2\, T}{2\sigma^2} \left( 
\frac1{\sqrt{a}}-\frac1{\sqrt{b}}
\right)^2.
\]
\end{proposition}
The Nash equilibrium of the RCG is thus indeed in mixed strategies. The planner has to guess where the adversary is placing the risk types in the interval $[a,b]$. To make this as hard as possible for the planner, the adversary randomizes, either putting all risk types to the highest possible level of risk aversion or to the lowest possible level.\footnote{Notice that this strategy represents a probability distribution over probability distributions. With probability $p$, all mass is in $a$, otherwise it is in $b$. This is distinct from a situation where mass $p$ is in $a$ and mass $1-p$ is in $b$, i.e., where some agents have the highest risk type while others have the lowest one.} The planner reacts to this randomized strategy by picking a well-chosen middle ground. His optimal strategy is the decision that is for optimal $\gamma$ at the geometric mean of $a$ and $b$. 

\subsection{Robust Optimal Partitioning}

We next study what the robust planning problem looks like when the planner can offer agents a menu of $n$ possible choices, $m_1>\ldots>m_n$. Agents pick a choice from the menu by comparing their risk type to the partition implied by the $m_i$ as described in Lemma \ref{lemHM2}. We call the corresponding versions of our two games the $n$-ACG and the $n$-RCG.  In the case of the absolute criterion game $n$-ACG, we find that the argument of Lemma \ref{ACG} still applies. For any given decision menu, the adversary minimizes welfare by making agents as risk averse as possible, concentrating all mass in $b$. Having the possibility to offer multiple products does not help the planner here. The best he can do is to offer what is optimal for maximally risk averse agents with type $b$. 

Analyzing the relative criterion game $n$-RCG is more rewarding. Here, we do not attempt a full game-theoretic analysis  like in Proposition \ref{RCG}. Instead, we  focus on a simpler question, accounting for the fact that we are more interested in the planner's optimal behavior than in that of the adversary. In Proposition \ref{RCG}, the planner's robust optimal strategy is to implement the preferred decision of an agent whose risk type is the geometric mean of $a$ and $b$. In Proposition \ref{nRCG}, we extend this robust strategy to menus of decisions. We show that there is a  unique menu which generalizes the geometric mean strategy. 

\begin{proposition}\label{nRCG}
In any equilibrium of the $n$-RCG in which the planner plays a pure strategy, this strategy consists of offering the menu of choices $m^*_1>\ldots>m^*_n$ given by 
\[
m^*_i=\frac{\mu-r}{\sigma^2 \Gamma^*_i}
\]
where 
\begin{equation}\label{defGamma}
\Gamma^*_i=g^*(m^*_i)= \frac{a b}{h_{i-1}h_i}
\;\;\;\;\;\;\;\;\text{and}\;\;\;\;\;\;\;\;
h_i=\sqrt{a}\,\frac{i}{n}+\sqrt{b}\,\frac{n-i}{n}.
\end{equation}
The risk type $g^*_i$ of an agent who is indifferent between $m^*_i$ and $m^*_{i+1}$ is 
\begin{equation}\label{defg}
g^*_i = \frac{a b}{h_i^2}.
\end{equation}
The resulting candidate for an equilibrium value of $\mathcal{R}$ is given by  
\[
\mathcal{R}^*=- \,\frac{(\mu-r)^2\, T }{2\sigma^2 n^2} \left( 
\frac1{\sqrt{a}}-\frac1{\sqrt{b}}
\right)^2.
\]
\end{proposition}
For $n=1$, the result simplifies to $\Gamma^*_1=\Gamma^*_n=\sqrt{ab}$ as expected. The marginal risk types $g^*_i$ determine the partition into which agents sort themselves. Looking at the candidate for the planner's equilibrium utility loss $\mathcal{R}^*$, we see that it vanishes quadratically with $n$. Thus, already moderate values of $n$ substantially reduce the adversary's scope for harming the planner by picking an unfavorable distribution of risk types. Figure \ref{FIG:u} illustrates the robust strategy of Proposition \ref{nRCG} and the resulting partitions. Compared to the uniform distribution example in Figure \ref{FIG:example}, we see that the robust choices of strategies and partitions in the left panel are concentrated further to the left, i.e., there is a finer subdivision of the less risk averse types. Intuitively, the reason for this is that these types are more heterogeneous  in their preferences, i.e., the slope of the function $m^*(\gamma)$ is largest for small values of $\gamma$. Thus, an adversary who tries to create risk types whose preferences are not well-served by the current menu will tend to put more attention on less risk averse types. Conversely, we see in the right panel that the offered choice menus $m_i^*$ are more evenly spaced than in the uniform example of Figure \ref{FIG:example}.  

\begin{figure}[!h]
	\begin{subfigure}{.49\textwidth}
		\centering
		\includegraphics[width=0.9\linewidth]{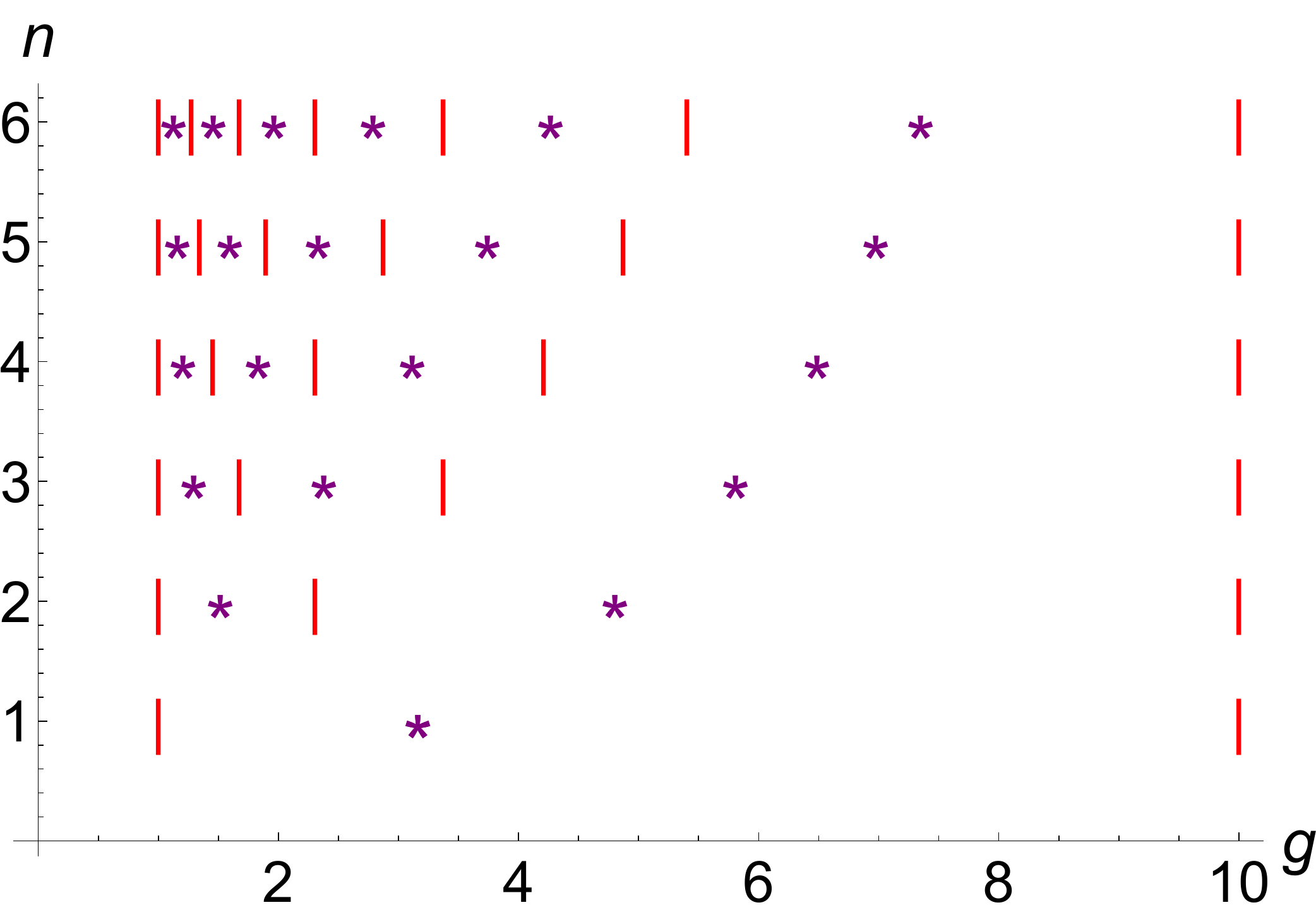}  
		\caption{$\Gamma_i^*$ (purple stars) and $g_i^*$ (red lines).}
		\label{FIG:gu}
	\end{subfigure}
\begin{subfigure}{.49\textwidth}
	\centering
	\includegraphics[width=0.9\linewidth]{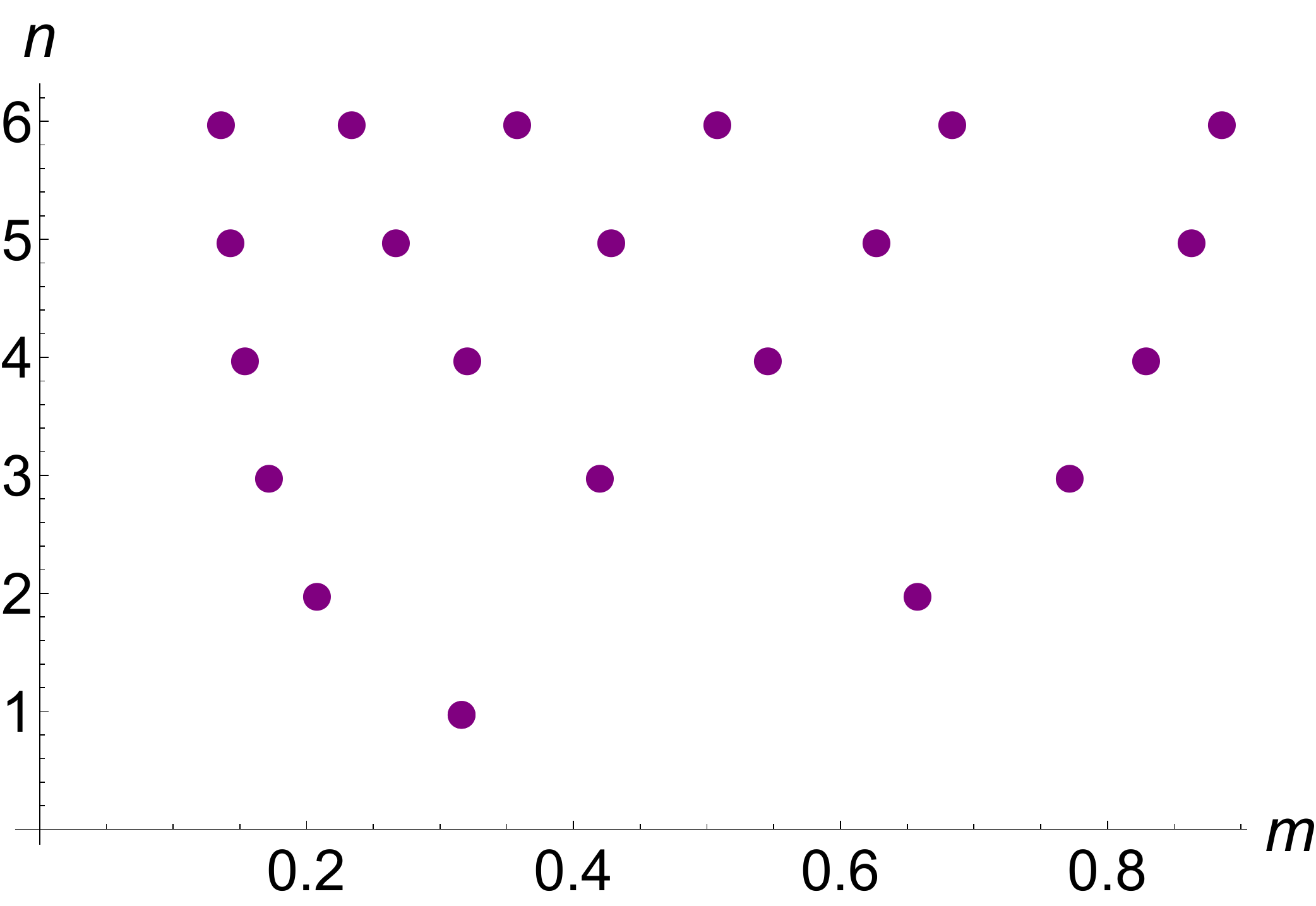}  
	\caption{$m_i^*$ (purple dots).}	
	\label{FIG:mu}
\end{subfigure}
	\caption*{ \footnotesize For varying $n$, the left panel shows the robust optimal targeted risk types $\Gamma_i^*$ and partition boundaries $g_i^*$  as given in Proposition \ref{nRCG}. The right panel shows the corresponding choice menus $m_i^*$. The parameters are $a=1$, $b=10$ and $(\mu-r)/\sigma^2=1$.}
	\caption{Robust strategies and partitions.}
	\label{FIG:u}
\end{figure}

\begin{remark}
Technically, the intuition behind the proposition is as follows. Suppose the planner would know that there are exactly $k\leq n$ risk types $\gamma_1,\ldots,\gamma_k$ that can arise from the adversary's strategy. Then the planner could  implement the menu $m^*(\gamma_1),\ldots, m^*(\gamma_k)$ and achieve $\mathcal{R}=0$, a perfect match between risk types and available choices. Thus, in order to be part of an equilibrium, the adversary's strategy must mix over at least $n+1$ different risk types so that the planner cannot offer a menu of perfect reactions. However, mixing over a set of risk types can only be optimal for the adversary if he is indifferent between them. This is the case if each of the risk types in the support of the adversary's strategy is a worst possible match for the decision menu offered by the planner. There are only $n+1$ candidate locations for such worst-possible matches. These candidates are the boundaries $a$ and $b$ and the points $g_i$ at which agents with the associated risk type are indifferent between adjacent strategies. From these considerations, we deduce the following \textit{indifference condition:} A menu of choices $m_1,\ldots,m_n$ can only be part of an equilibrium if the adversary is indifferent between the $n+1$ situations in which all agents have risk type $g_0,\ldots,g_n$. Here, the $g_i$ are the partition boundaries induced by the menu choices $m_i$ via the harmonic mean condition, augmented by $g_0=a$ and $g_n=b$. In the proof, we show that this indifference condition uniquely pins down the menu given in the proposition.\end{remark}

\begin{remark}
While we have not explicitly constructed an equilibrium, the proof of the proposition suggests what it would have to look like, giving some insight into possible strategies of the adversary. The numbers $\Gamma^*_i$ that determine the planner's strategy are chosen such that the adversary is indifferent between distributions of risk types that are concentrated in any of the numbers $g^*_i$, including the interval boundaries $a$ and $b$. He prefers these distributions over all others and can thus mix over them, randomly placing risk types in these locations such that the planner's strategy becomes a best response.\footnote{There is a small subtlety here. Agents with risk type $g^*_i$ are indifferent between $m^*_i$ and $m^*_{i+1}$. In order to stabilize the planner's behavior in equilibrium, the adversary needs to control the fraction of agents with type $g^*_i$ which pick either of these options. One can think of this as the adversary placing distinct atoms on $g_i^*-\varepsilon$ and $g_i^*+\varepsilon$. That tie-breaking rules need to be specified to ensure existence of equilibria is fairly common in games with continuous action space, see \cite{simon1990discontinuous}.} 
\end{remark}

Stepping outside the game-theoretic setting, we can also read Proposition \ref{nRCG} as a distribution-independent performance bound. As long as the planner follows the prescribed strategy, and as long as risk types are distributed within the interval $[a,b]$, the relative loss criterion $\mathcal{R}$ is bounded from below by the number $\mathcal{R}^*$ given in the proposition. The proposition thus gives a bound on the welfare loss from not knowing the distribution precisely when implementing a decision menu with $n$ choices. Moreover, it gives an explicit menu which achieves this bound. Combining Propositions \ref{bd2} and \ref{nRCG}, one can extend this to a bound which quantifies the welfare loss from implementing the robust $n$-element decision menu of Proposition \ref{nRCG} rather than fully personalized decisions. 

\paragraph{Comparative Statics.} We close this section with some further observations about the comparative statics of the partitions described by the numbers $g_i^*$ and $\Gamma_i^*$ defined in \eqref{defGamma} and \eqref{defg}. Clearly, when $n$ increases, the numbers $g_i^*$ and $\Gamma_i^*$ decrease as the partition becomes finer. For any fixed $n$, there exist increasing functions $\varphi_i$ and $\overline{\varphi}_i$ such that
\[
g_i^*= a \varphi_i\left(\frac{b}{a}\right)
\;\;\text{ and }\;\;
\Gamma_i^*= a \overline{\varphi}_i\left(\frac{b}{a}\right).
\]
Consequently, when $a$ and $b$ are multiplied by the same factor, the numbers $g_i^*$ and 
$\Gamma_i^*$ are multiplied by that factor as well. We next study the relative location 
\[
r^*_i(a,b)=\frac{g_i^*-a}{b-a}
\]
of $g_i^*$ within the interval $[a,b]$. The relative location $r^*_i$ lies between $0$ and $1$ and captures how much of the total distance between $a$ and $b$ lies between $a$ and $g_i^*$.
\begin{corollary}\label{corriab}
For fixed $n$ and $0< i < n$, the function $r^*_i(a,b)$ has the following properties:
\item[(i)] For any $\lambda>0$, $r^*_i(\lambda a,\lambda b)=r^*_i(a,b)$.
\item[(ii)] $\lim_{b\uparrow \infty} r^*_i(a,b)=0$.
\item[(iii)] $\lim_{a\downarrow 0} r^*_i(a,b)=0$.
\item[(iv)] $\lim_{b\downarrow a} r^*_i(a,b)=\frac{i}{n}$.
\end{corollary}
By definition $r^*_0(a,b)=0$ and $r^*_n(a,b)=1$ as $g_0^*=a$ and $g_n^*=b$. 
Property (i) reflects again the fact that scaling both $a$ and $b$ by the same factor just scales up the entire partition. Properties (ii) and (iii) consider situations where the heterogeneity in possible risk types $b/a$ goes to infinity, either because $b$ gets large or because $a$ gets small. In this case, the entire partition gets more and more concentrated at the lower, least risk averse type $a$.
In case (iii), the limiting partition for $a=0$ is degenerate  with all  boundaries except for $g_n^*=b$ converging to $a$. In contrast, in case (ii), partition boundaries $g_i^*$ converge to $\infty$ with $b$ but do so very slowly so that, in relative terms, they move closer to the fixed lower boundary $a$. Finally, in case (iv) where heterogeneity vanishes, $b\approx a$, we see that  
$r^*_i(a,b)$ converges to $i/n$, corresponding to an evenly spaced partition of the interval $[a,b]$. Analogous results hold for the relative locations $\rho^*_i(a,b)$ of the targeted risk types $\Gamma_i^*$, 
\[
\rho^*_i(a,b)=\frac{\Gamma_i^*-a}{b-a}.
\]
There are only two small differences compared to Corollary \ref{corriab}. First, $\rho^*_n(a,b)$ also converges to $0$ when $b$ goes to infinity or when $a$ goes to zero. Second, one can show that
$$\lim_{b\downarrow a} \rho^*_i(a,b)=\frac{i-\frac{1}2}{n}.$$
Thus, in the limit of vanishing heterogeneity, the numbers $\Gamma_i^*$ and $g_i^*$ together form an evenly spaced partition of the interval $[a,b]$ into $2n$ subintervals of length $1/(2n)$.

\section{Dynamic Investment with Multiple Assets}\label{secdynamic}
So far, our analysis of the planner's  decisions was largely a static one. In this section, we  explain how to embed it into a dynamic investment problem in the spirit of the classical Merton problem.\footnote{See \cite{merton1971optimum} for the origin and \cite{rogers2013optimal} for a recent textbook treatment.} One celebrated result in that setting is the two-fund separation theorem: In a market where all assets are geometric Brownian motions, all agents regardless of their risk preferences optimally split their investments between the risk-free asset and a fixed portfolio of the risky assets which is identical for all agents, the so-called tangency portfolio. We will show that in our setting, the optimality of two fund separation is inherited by the planner's preferences. Consequently, our previous analysis can be interpreted in the context of a multi-asset Merton investment problem. The univariate parameter $m$ becomes the fraction of wealth that is invested into the tangency portfolio. 

Throughout this section, we simplify the analysis by restricting attention to deterministic, time-dependent strategies that prescribe at every instant a fraction of wealth that is invested into the risky asset. We discuss this assumption further in the final part of this section, touching also upon the issue of time-inconsistency of the planner's preferences. 

\paragraph{Dynamic investment with a single asset.} We begin with the case of investment into a single risky asset $S$ over the time horizon $[0,T]$. The asset follows a geometric Brownian motion with drift $\mu$ and volatility $\sigma$ so that its evolution can be described by the stochastic differential equation (SDE)
\[
dS_t =\mu S_t dt +\sigma S_t dW_t
\]
where $W_t$ is a standard Brownian motion. Besides the risky asset, there is a risk-free asset with constant rate of return $r<\mu$. We denote by $V_t$ the wealth process that arises from investing at time $t$ a deterministic time-dependent fraction $m_t$ of wealth into the risky asset and the remainder into the risk-free asset. Its evolution is described by the SDE 
\[
dV_t= (r +m_t(\mu-r)) V_t dt +m_t\sigma V_t dW_t.
\]
With an initial wealth of $V_0=1$, it follows that wealth at time $T$ is given by
\begin{equation}\label{VT}
V_T=\exp\left( 
rT + (\mu-r) \int_0^T m_t  dt - \frac{ 1}{2}\sigma^2\int_0^T m_t^2 dt+ \sigma\int_{0}^T m_t  dW_t
\right).
\end{equation}
Consider a planner who chooses the strategy $m_t$, facing a population of power utility agents like in our static baseline model. We will argue below that it is optimal for such a planner to choose only between time-constant strategies $m$. With a time-constant $m$, terminal wealth $V_T$ can then be written as  
\[
V_T=\exp\left( 
rT + (\mu-r) m T - \frac{ 1}{2}\sigma^2 m^2 T+ \sigma m  W_T
\right).
\]
Since $W_T$ is normally distributed with mean $0$ and variance $T$, it follows that $V_T$ has the same distribution and same dependence on $m$ as the outcome quantity $R(m,Z)$ in our static baseline model. Consequently, the entire analysis of the static problem applies to the dynamic problem as well. To justify our focus on time-constant strategies, we consider the certainty equivalent of $V_T$ given in \eqref{VT} for an agent with risk type $\gamma>0$,
\[
u_\gamma^{-1}(E[u_\gamma(V_T)])=
\exp\left(
r T + (\mu-r) \int_0^T m_t  dt - \frac12  \sigma^2 \gamma \int_0^T m_t^2 dt
\right).
\]
Consider some strategy $(m_t)$ with an average investment fraction of $\kappa=\frac{1}{T}\int_0^T m_t  dt$. One can easily show that the constant strategy $m_t\equiv \kappa$ minimizes $\int_0^T m_t^2 dt$ among all strategies with average investment fraction $\kappa$. It follows that $m_t\equiv \kappa$ maximizes $u_\gamma^{-1}(E[u_\gamma(V_T)])$ among all 
strategies with average investment fraction $\kappa$. This holds regardless of the value of $\gamma$, i.e., given a fixed average investment fraction, all risk types agree on the best possible strategy and prefer the time-constant one. The time-constant strategy Pareto dominates all other strategies with the same investment fraction. Consequently, the planner can restrict attention to strategies which prescribe constant investment fractions over time. This shows that the dynamic problem can be reduced to a static one. 

\paragraph{Multiple assets.} We next argue that the investment problem with $d$ risky assets can also be reduced to the univariate static problem with payoff \eqref{RmZ}. Again, the basic argument is to rule out strategies that are dominated from the perspective of all risk types. We show that among all strategies that reach a given rate of return, all risk types prefer a strategy that is a multiple of the tangency portfolio. Thus, by Pareto dominance, the multi-asset investment problem can be reduced to a single asset problem where the single asset is the tangency portfolio. We assume that our $\mathbb{R}^d$-valued stock price process $S_t$ follows the SDE 
\[
dS_t = \textnormal{diag}(S_t)(\bar{\mu} dt + \bar{\sigma} dW_t).
\]
Here, $\bar{\mu}$ is a vector in $\mathbb{R}^d$ with $\bar{\mu}_i>r$, $\bar{\sigma}\in \mathbb{R}^{d\times d}$ is an invertible matrix, $W_t$ is a $d$-dimensional standard Brownian motion, and $\textnormal{diag}(S_t)$ denotes the $d\times d$
diagonal matrix with diagonal entries $S_t$. We denote by $\bar{m}$ a (time-constant) $d$-vector of fractions of wealth invested into the risky assets. The implied dynamics of the wealth process is given by 
\[
dV_t= V_t(r+ \bar{m}^\top (\bar{\mu}-r\iota)   )dt +V_t \bar{m}^\top  \bar{\sigma} dW_t
\]
where $\iota$ denotes the all-ones vector in $\mathbb{R}^d$ and $\top$ denotes matrix transposition. The certainty equivalent of an agent with risk type $\gamma$ is then given by
\[
u_\gamma^{-1}(E[u_\gamma(V_T)])=
\exp\left(
r T +   \bar{m}^\top(\bar{\mu}-r\iota) T - \frac12  \gamma \bar{m}^\top \bar{\sigma}\bar{\sigma}^\top \bar{m} T
\right).
\]
Solving the linear-quadratic problem in the exponent, 
\begin{equation}\label{logce}
\sup_{\bar{m}} \bar{m}^\top(\bar{\mu}-r\iota)  - \frac12  \gamma \bar{m}^\top \bar{\sigma}\bar{\sigma}^\top \bar{m},
\end{equation}
it follows that the individually optimal strategy of an agent with risk type $\gamma$ is given by 
\[
\bar{m}^*(\gamma)=\frac{1}{\gamma}\bar{m}^\tau \;\;\text{ where }\;\; \bar{m}^\tau=(\bar{\sigma}\bar{\sigma}^\top)^{-1}(\bar{\mu}-r\iota).
\]
The vector $\bar{m}^\tau$ is called the tangency portfolio. We will show that to solve the planner's problem it is sufficient to consider strategies which are multiples of the tangency portfolio, i.e., to restrict attention to vectors of the form
\[
\bar{m} = c\,\bar{m}^\tau = c (\bar{\sigma}\bar{\sigma}^\top)^{-1}(\bar{\mu}-r\iota)
\]
where $c$ is a positive scalar. From this claim, it follows that we can write
\[
u_\gamma^{-1}(E[u_\gamma(V_T)])=
\exp\left(
r T +   c\, k T - \frac12  \gamma c^2 k T
\right)
\]
where $k$ is given by 
\[
k=(\bar{\mu}-r\iota)^\top(\bar{\sigma}\bar{\sigma}^\top)^{-1}(\bar{\mu}-r\iota)>0.
\]
Thus, finding the optimal scalar $c$ is equivalent to finding the optimal investment fraction $m$ in the single asset case, i.e., the multi-asset problem collapses to the problem of Lemma \ref{lem1} with $m=c$, $\mu-r=k$ and $\sigma^2=k$. To show that we can restrict attention to multiples of $\bar{m}^\tau$, we consider the constrained maximization of the log-certainty equivalent
\begin{equation}\label{logcek}
\sup_{\bar{m}} \bar{m}^\top(\bar{\mu}-r\iota)  - \frac12  \gamma \bar{m}^\top \bar{\sigma}\bar{\sigma}^\top \bar{m}\;\; \text{s.t.}\;\;\bar{m}^\top(\bar{\mu}-r\iota)=k 
\end{equation}
for some positive $k$. Solving this problem by Lagrangian optimization boils down to subtracting a multiple $\Lambda$ of 
the first term $\bar{m}^\top(\bar{\mu}-r\iota)$ from the objective,
\[\sup_{\bar{m}} (1-\Lambda )\bar{m}^\top(\bar{\mu}-r\iota)  - \frac12  \gamma \bar{m}^\top \bar{\sigma}\bar{\sigma}^\top \bar{m} 
\]
Varying the Lagrange multiplier $\Lambda$ is thus equivalent to varying the risk aversion coefficient $\gamma$ in \eqref{logce}, the unconstrained version of \eqref{logcek}. In particular, since the solution to \eqref{logce} is a multiple of $\bar{m}^\tau$ for any $\gamma$, this property is inherited by the constrained version \eqref{logcek}. Among strategies with a fixed rate of return $k>0$, all risk types prefer the multiple of the tangency portfolio with return $k$ over all alternative strategies. Thus, multiples of the tangency portfolio are Pareto dominant and the planner can restrict attention to them.  

\paragraph{Initial Wealth.}

From a finance perspective, our assumption of unit initial wealth may seem restrictive. However, it can easily be relaxed when the planner has a power utility function with parameter $\eta$. Denote by $V_0(\gamma)>0$ the (total) initial wealth of agents with risk type $\gamma\in [a,b]$. Then we can write the certainty equivalent for risk type $\gamma$ as
\[
\CE(\gamma, m, V_0(\gamma))=u_{\gamma}^{-1}(E[u_{\gamma}(V_0(\gamma) R(m,Z))])=
 V_0(\gamma) \CE(\gamma, m) 
\]
where $\CE(\gamma, m)$ is the certainty equivalent with unit initial wealth as before. For logarithmic utility of the planner, $\eta=1$, it follows immediately that 
\[
\EE[\log(\CE(\gamma, m(\gamma), V_0(\gamma) )])= \EE[\log(\CE(\gamma, m(\gamma))])] +
\EE[\log(V_0(\gamma))]). 
\]
Since the first summand does not depend on $V_0$ and the second summand does not depend on $m$, the planner's optimization problem is not affected by the distribution of initial wealth. A planner with logarithmic utility just optimizes the population average of the certainty equivalent growth rate without taking into account how wealth varies with risk type. For a power utility planner with $\eta\neq 1$, the objective of maximizing
\[
\EE[v(\CE(\gamma, m(\gamma), V_0(\gamma)))]
\]
is, up to increasing linear transformations, equivalent to maximizing
\[
\EE[V_0(\gamma)^{1-\eta} v(\CE(\gamma, m(\gamma)))] \;\;\;\;\text{and}\;\;\;\;
\widetilde{\EE}[v(\CE(\gamma, m(\gamma)))]
\]
where $\widetilde{\EE}$ is an expected value with respect to a reweighted density
\[
\widetilde{f}(g)=\frac{1}{\EE[V_0(\gamma)^{1-\eta}]}\,V_0(g)^{1-\eta} f(g).
\]
Thus, up to a reweighting of $f$, our   analysis also applies with heterogeneous initial wealth. 

It is instructive to study the distortion that occurs in going from  $f$ to $\widetilde{f}$. For $\eta=0$, the inequality-neutral planner applies a simple weighting by initial wealth. For $\eta \in (0,1)$ agents with larger initial wealth still receive a larger weight in the planner's objective. The logarithmic planner, $\eta=1$, applies no distortion $f \equiv \widetilde{f}$. Finally, for $\eta >1$ the planner is so inequality averse that he aims at redistribution, giving more weight to the risk preferences of poorer types and less weight to types who already have a lot of money.

\paragraph{Time-Consistency.}

In Lemmas \ref{OS1} and \ref{OS2}, we saw that optimal decisions in our problem may depend on the length of the time horizon $T$. Consequently, the planner's problem is not time-consistent in general. If the planner reconsiders his decision at a later date, with a shorter remaining time horizon $T$, he will prefer a different choice of $m$. Thus, at each time point, the planner wishes to commit on a time-constant $m$ for the remaining planning horizon -- but the optimal level of $m$ evolves over time. 

The optimal decisions we characterize are thus pre-commitment strategies in the sense of \cite{strotz1955myopia}. They are only feasible if the planner has the power to commit on sticking with his decisions and not revising them. Besides the pre-commitment strategy, \cite{strotz1955myopia} also introduces the notion of a sophisticated strategy under which the planner optimizes his current objective taking into account that his future selves will do the same. For the special case of a linear $v$, $\eta=0$, such sophisticated strategies have been derived  by \cite{DesmettreSteffensen} for a collective investment problem similar to ours. 

Both the pre-commitment and the sophisticated solution have their merits, depending on the context and, in particular, on the plausibility of the commitment assumption. Given that we mostly think of our planner as acting on behalf of others, assuming that the planner can credibly commit on a certain investment strategy may be comparatively plausible. For instance, the strategy might be formalized in a contract that the planner makes with the agents at the beginning of the investment horizon. 

When we consider logarithmic utility for the planner like in Sections \ref{secbounds} and \ref{secrobust}, all complications of time-inconsistency vanish because the optimal decisions $m$ do not depend on $T$. In this case, the pre-commitment and sophisticated solutions coincide. An important consequence is that, intuitively, the restriction to deterministic strategies is also without loss of generality then by classical arguments:  The sophisticated strategy can be computed backwards in time by dynamic programming. At every instant, current wealth is merely a multiplicative factor which does not influence optimal investment due to the power utilities of the agents and the planner. Thus, the optimal sophisticated strategy will be deterministic. When the time-consistency problems disappear like in the logarithmic case, this property is inherited by our pre-commitment strategy.

\appendix
\section{Proofs}

\begin{proof}[Proof of Lemma \ref{lem1}]
Since $u$ is a power utility function, we can write
\begin{align*}
&u^{-1}(E[u(R(m,Z))])
=\exp\left(rT+(\mu-r)mT-\frac{1}2 \sigma^2 m^2T\right) E\left[
\exp\left(
m\sigma Z \sqrt{T} (1-\gamma)
\right)
\right]^{\frac{1}{1-\gamma}}\nonumber
\end{align*}
pulling a deterministic factor outside of the certainty equivalent. Since $Z$ is standard normal, we know that $E[\exp(\theta Z)]=\exp(\theta^2/2)$ for any  $\theta$ and thus\[
E\left[
\exp\left(
m\sigma Z \sqrt{T} (1-\gamma)
\right)
\right]^{\frac{1}{1-\gamma}} = \exp\left(\frac12m^2\sigma^2 T(1-\gamma) \right).
\]
This is the claimed formula for the certainty equivalent. As a monotonic transformation of a quadratic polynomial, it has a unique maximizer in $m$ as stated in the lemma. 
\end{proof}

\begin{proof}[Proof of Lemma \ref{OS1}]
The planner maximizes the smooth function $O(m)= \EE[v(\CE(\gamma,m))]$. Taking the derivative with respect to $m$ yields
\begin{align*}
O'(m)&= \EE\left[v'(\CE(\gamma,m))\CE(\gamma,m) \left(\mu-r-m\sigma^2 \gamma  \right)T\right]
= \EE\left[h(\gamma,m) \left(\mu-r-m\sigma^2 \gamma  \right)\right]T
\end{align*}
where the function $h$ is positive by our assumptions on $v$. The first order condition $O'(m)=0$ can thus be written as $m=\Phi(m)$ where
\[
\Phi(m)=\frac{\mu-r}{\sigma^2 \Gamma(m)} \;\;\text{ with }\;\;\Gamma(m)=\frac{\EE[\gamma h(\gamma,m)]}{\EE[h(\gamma,m)]}.
\]
Since $\Gamma(m)$ is the expected value of $\gamma$ after a change of measure which preserves the support $[a,b]$, we have $\Gamma(m) \in [a,b]$ and thus $\Phi(m) \in [m^*(b),m^*(a) ]$ for all $m$. Since $m$ can take any positive value, the equation $m=\Phi(m)$ must thus have at least one solution and all solutions must lie in the interval $[m^*(b),m^*(a) ]$. Moreover, $m<\Phi(m)$ for sufficiently small $m$, $m< m^*(b)$, and $m>\Phi(m)$ for sufficiently large $m$, $m> m^*(a)$. This implies $O'(m)>0$ for small $m$ and $O'(m)<0$ for large $m$. Since it is smooth by our assumption, the function $O$ must thus attain an interior global maximum somewhere in the interval $[m^*(b),m^*(a) ]$ and that maximum must satisfy the first order condition $O'(m)=0$.
\end{proof}

\begin{proof}[Proof of Lemma \ref{OS2}]
Compared to the situation in Lemma \ref{OS1}, we now have an explicit utility function which implies an explicit change of measure, $h(\gamma,m) = \CE(\gamma,m)^{1-\eta}$. For $\eta=1$, we thus get  $h(\gamma,m)=1$, implying that $\Gamma=\EE[\gamma]$ does not depend on $m$. The function $\Phi(m)$ is thus constant and intersects the identity function $m$ exactly once. This proves (i). The formulation of the first order condition in the lemma, follows after noting that by $h(\gamma,m) = \CE(\gamma,m)^{1-\eta}$ and Lemma \ref{lem1}
\[
\frac{h(\gamma,m)}{\EE[h(\gamma,m)]} = \frac{\exp(\gamma \theta(m))}{\EE[\exp(\gamma \theta(m))]}
\]
where $\theta(m)=\frac{1}{2}\sigma^2(\eta-1)T m^2$. To conclude the proof, we rely on the fact that the function $\psi:\mathbb{R} \rightarrow [a,b]$
\[
\psi(t)= \EE\left[ \gamma \; \frac{\exp(\gamma t)}{\EE[\exp(\gamma t)]} \right]
\]
is increasing in $t$ with $\psi(0)= \EE[\gamma]$. To see this, note that 
\[
\psi'(t)= \EE\left[ \gamma^2 \; \frac{\exp(\gamma t)}{\EE[\exp(\gamma t)]} \right] - \EE\left[ \gamma \; \frac{\exp(\gamma t)}{\EE[\exp(\gamma t)]} \right]^2
\]
is positive as it is the variance of $\gamma$ after a change of measure proportional to $\exp(\gamma t)$. We now write
\[
\Phi(m)=\frac{\mu-r}{\sigma^2 \psi( \theta(m) )}
\]
and note that for $\eta <1$ the function $\theta(m)$ is decreasing with $\theta(0)=0$. It follows that $\psi( \theta(m) ) \leq \EE[\gamma]$ so $\Phi(m) \in [m^*(\EE[\gamma]), m^*(a)]$. This shows (iii). The converse argument, using that for $\eta >1$ the function $\theta(m)$ is increasing with $\theta(0)=0$, shows most of (ii). It remains to argue that the equation $m=\Phi(m)$ has a unique solution in this case. To this end, note that $\psi(\theta(m))$ is now increasing, so $\Phi(m)$ is decreasing. Since the decreasing function $\Phi(m)$ can intersect the increasing identity function only once, it follows that $m=\Phi(m)$ has a unique solution. 
\end{proof}

\begin{proof}[Proof of Lemma \ref{lemHM1}]
We can write the planner's objective as
\[
O(g_0,\ldots,g_n)=\sum_{i=1}^n U(g_{i-1},g_{i},m^*(g_{i-1},g_i))
\]
where 
\begin{equation}\label{defU}
U(\alpha,\beta,m)=\int_\alpha^\beta v(\CE(g,m)) f(g) dg
\end{equation}
for $\alpha,\beta \in [a,b]$ and $m\in \mathbb{R}$ and where the $g_i$ satisfy $a=g_0< g_1\ldots<g_n=b$. We prove the lemma by showing that the harmonic mean condition is equivalent to the first order condition $\frac{\partial O}{\partial g_i } =0$ for all $i$ with $0<i<n$. To this end, note first that due to the optimality of $m^*$ the partial derivatives with respect to $m$ vanish,
\[
\frac{\partial U(g_{i-1},g_i, m^*(g_{i-1},g_i) )}{\partial m} =0.
\]
We can thus write our first order condition as
\[
0= \frac{\partial O(g_0,\ldots,g_n)}{\partial g_i }
= \frac{\partial U(g_{i-1},g_i, m^*(g_{i-1},g_i) )}{\partial \beta} + \frac{\partial U(g_i,g_{i-1}, m^*(g_{i},g_{i+1}) )}{\partial \alpha}.
\]
By \eqref{defU} and the monotonicity of $v$, this condition is the same as 
\begin{equation}\label{CECEcond}
0 = \CE(g_i,m^*(g_{i-1},g_{i})) f(g_i)  - \CE(g_i,m^*(g_{i},g_{i+1})) f(g_i). 
\end{equation}
With $m_{i-1}=m^*(g_{i-1},g_{i})$ and  $m_{i}=m^*(g_{i},g_{i+1})$, this condition becomes,
by  Lemma \ref{lem1}, 
\[
(\mu-r) m_{i-1}-\frac{1}{2} g_i m_{i-1}^2 \sigma^2 = (\mu-r) m_{i}-\frac{1}{2} g_i m_{i}^2 \sigma^2.
\]
Plugging in 
$$m_{i-1}=\frac{\mu-r}{\sigma^2 g^*(m_{i-1})}\;\; \text{ and }\;\;
m_{i}=\frac{\mu-r}{\sigma^2 g^*(m_{i})},$$ 
 this condition can be rewritten into 
\[
\frac{1}{g^*(m_{i-1})}- \frac{1}{2} \frac{g_i}{g^*(m_{i-1})^2}.
=
\frac{1}{g^*(m_i)}- \frac{1}{2} \frac{g_i}{g^*(m_i)^2}.
\]
Solving this equation for $g_i$ and simplifying gives the harmonic mean condition
\[
g_i = \frac{2}{\frac{1}{g^*(m_i)} + \frac{1}{g^*(m_{i-1})}}
\]
\end{proof}

\begin{proof}[Proof of Lemma \ref{lemHM2}]
The proof of the lemma is contained in the one of Lemma  \ref{lemHM1}. It suffices to note that \eqref{CECEcond} is equivalent to the indifference condition of risk types at the boundary,
$$\CE(g_i,m^*(g_{i-1},g_{i})) = \CE(g_i,m^*(g_{i},g_{i+1})).$$
\end{proof}

\begin{proof}[Proof of Lemma \ref{lemCEgrowth}]
Plugging 
\[
m(\gamma)=\frac{\mu-r}{\sigma^2 G(\gamma)}
\]
into 
\[
\frac{1}{T}\log(\CE(\gamma,m(\gamma)))=r+m(\gamma)(\mu-r)-\frac{1}{2}\gamma m(\gamma)^2\sigma^2 
\]
and applying $\EE$ yields
\[
\frac{1}{T}\EE\left[\log(\CE(\gamma,m(\gamma)))\right]=r+\frac12 \left(\frac{\mu-r}{\sigma}\right)^2  \EE\left[ \frac{2}{G(\gamma)}-\frac{\gamma}{G(\gamma)^2}
\right]
\]
as claimed. 
\end{proof}

\begin{proof}[Proof of Lemma \ref{lemRepr}]
For a given $n$-element partition $a=g_0<\ldots < g_n=b$ we define the function $\overline{G}(\gamma)$ via $\overline{G}(\gamma)=M_i/P_i$ for $\gamma \in [g_{i-1},g_i)$ where
\[
M_i= \int_{g_{i-1}}^{g_{i}} g f(g) dg
\;\;\text{ and }\;\;
P_i= \int_{g_{i-1}}^{g_{i}} f(g) dg
\]
Thus, in line with the logarithmic utility case in Lemma \ref{OS2}, we set the decision that is applied for risk types in $[g_{i-1},g_i)$ equal to the optimal decision for the mean risk type in the interval. Since maximizing $E_n^*$ is equivalent to maximizing the planner's objective, we know that $E_n^*$ can be written as
\[
E_n^*\; = \sup_{a=g_0<\ldots < g_n=b } \; \EE\left[\frac{2}{\overline{G}(\gamma)} -\frac{\gamma}{\overline{G}(\gamma)^2} \right]. 
\]
We complete the proof by showing that for any fixed partition
\begin{equation}\label{PiMieq}
\EE\left[\frac{2}{\overline{G}(\gamma)} -\frac{\gamma}{\overline{G}(\gamma)^2} \right]
=\sum_{i=1}^n \frac{P_i^2}{M_i}.
\end{equation}
Since $P_i=\PP\left(\gamma \in [g_{i-1},g_{i}] \right)$ and 
$\EE\left[\gamma \left| \gamma \in [g_{i-1},g_{i}]\right.\right]=M_i/P_i$, \eqref{PiMieq} immediately implies \eqref{eqRepr}. To see \eqref{PiMieq}, we plug in the definition of 
$\overline{G}$ on the left hand side to obtain
\begin{align*}
\EE\left[\frac{2}{\overline{G}(\gamma)} -\frac{\gamma}{\overline{G}(\gamma)^2} \right]
&=\sum_{i=1}^n \int_{g_{i-1}}^{g_i} \left(\frac{2 P_i}{M_i}-\frac{P_i^2}{M_i^2} \, g \right)f(g) dg
=\sum_{i=1}^n 2\frac{P_i^2}{M_i}-\frac{P_i^2}{M_i}
\end{align*}
using the linearity of the integral and the definitions of $P_i$ and $M_i$.
\end{proof}

\begin{proof}[Proof of Lemma \ref{bd1} and Remark \ref{rembd1}]
Denote by $\varphi(g)=\frac{a+b-g}{ab}$ the linear function which connects the points $(a,1/a)$ and $(b,1/b)$. Since the map $g \mapsto 1/g$ is convex, we have $1/g \leq \varphi(g)$ for all $g \in [a,b]$. In particular, since $\gamma$ has support $[a,b]$ we have the upper bound
\[
\EE\left[\frac{1}{\gamma}\right]\leq \EE\left[\varphi(\gamma)\right] 
=   \frac{(a+b-\EE[\gamma]) \EE[\gamma]}{ab\; \EE[\gamma]}\leq 
\frac{(a+b)^2}{4 ab}\frac{1}{\EE[\gamma]} 
=\frac{\frac{a}{b}+\frac{b}{a}+2}{4} \frac{1}{\EE[\gamma]} 
\]
where the second inequality uses that the expression $(a+b-z)z$ is a quadratic polynomial in $z$ which is maximal for $z=(a+b)/2$. Replacing $\EE[\gamma]$ in the numerator of the fraction by this maximizer gives the upper bound. Finally, to see the claim in Remark \ref{rembd1}, note that the first inequality is sharp if $\gamma$ takes only the two values $a$ and $b$ and that the second inequality is sharp if $\EE[\gamma]=(a+b)/2$. Thus, the inequality becomes an equality iff $\gamma$ takes values $a$ and $b$ with equal probability. 
\end{proof}

\begin{proof}[Proof of Proposition \ref{bd2}]
Denote by $\bar{g}_0,\ldots, \bar{g}_n$ the geometric partition of $[a,b]$. This partition is defined by $\bar{g}_0=a$ and $\bar{g}_i=(b/a)^{1/n} \bar{g}_{i-1}$.  Our goal is to show 
\begin{equation}\label{bdgoal}
E^*_\infty= \EE \left[\frac{1}{\gamma}\right] \leq C_n \sum_{i=1}^n \frac{\PP\left(\gamma \in [\bar{g}_{i-1},\bar{g}_{i}] \right)}{\EE\left[\gamma \left| \gamma \in [\bar{g}_{i-1},\bar{g}_{i}]\right.\right]}
\end{equation}
with 
\[
C_n=\frac{\left(\frac{b}{a}\right)^{\frac1n}+\left(\frac{a}{b}\right)^{\frac1n}+2 }{4}
\]
for this particular partition. The desired inequality then follows from 
\[
\sum_{i=1}^n \frac{\PP\left(\gamma \in [\bar{g}_{i-1},\bar{g}_{i}] \right)}{\EE\left[\gamma \left| \gamma \in [\bar{g}_{i-1},\bar{g}_{i}]\right.\right]} 
\leq 
\sup_{a=g_0<\ldots < g_n=b
} \;
\sum_{i=1}^n \frac{\PP\left(\gamma \in [g_{i-1},g_{i}] \right)}{\EE\left[\gamma \left| \gamma \in [g_{i-1},g_{i}]\right.\right]} =E_n^*.
\]
Note that the geometric partition has the property that for 
 all of the intervals $[\bar{g}_{i-1},\bar{g}_{i}]$ the ratio of lower and upper interval boundary is $(b/a)^{1/n}$. This implies that we can apply Lemma \ref{bd1} to the distribution of $\gamma$ conditional on $\gamma \in [\bar{g}_{i-1},\bar{g}_{i}]$
and obtain 
\begin{align}\label{bd1bd}
\EE\left[\left. \frac{1}{\gamma}\right| \gamma \in [\bar{\gamma}_{i-1},\bar{\gamma}_{i}] \right] \leq C_n 
\frac{1}{\EE\left[\left. {\gamma}\right| \gamma \in [\bar{\gamma}_{i-1},\bar{\gamma}_{i}] \right] }.
\end{align}
To show \eqref{bdgoal}, we thus apply the law of iterated expectations and then \eqref{bd1bd},
\begin{align*}
\EE \left[\frac{1}{\gamma}\right] 
&=\sum_{i=1}^n 
\PP\left(\gamma \in [\bar{g}_{i-1},\bar{g}_{i}] \right)
\EE\left[\left. \frac1{\gamma} \right| \gamma \in [\bar{g}_{i-1},\bar{g}_{i}]\right]
\leq C_n \sum_{i=1}^n \frac{\PP\left(\gamma \in [\bar{g}_{i-1},\bar{g}_{i}] \right)}{\EE\left[\gamma \left| \gamma \in [\bar{g}_{i-1},\bar{g}_{i}]\right.\right]}.
\end{align*}
\end{proof}

\begin{proof}[Proof of Corollary \ref{corbd}]
Using that $a/b\leq 1$, we obtain from Proposition \ref{bd2} the inequality $E_\infty^* \leq R_n E_n^*$ where $R_n  = \frac14(\left(\frac{b}a\right)^\frac1n+3)$. 
This implies $E_\infty^* \leq R E_n^*$ for all $R\geq R_n$. Solving the condition $R\geq R_n$ for $n$ gives the desired condition on $n$, where we note that $R\geq 1$ implies $\log(4R-3) \geq 0$. \end{proof}

\begin{proof}[Proof of Lemma \ref{ACG}]
Using Lemma \ref{lem1}, we can write
\[
\mathcal{A}(m,F)=rT+(\mu-r)m T-\frac12 m^2 \sigma^2 T \EE_F[\gamma].
\]
Since $\EE_F[\gamma]\in [a,b]$, $\mathcal{A}$ is quadratic in $m$ and has a unique maximum. 
Thus, the planner's best response to any strategy of the adversary is to choose the pure strategy
\[
m^*(\mathbb{E}[\EE_F[\gamma]])=\frac{\mu-r}{\sigma^2 \mathbb{E}[\EE_F[\gamma]]}
\]
where $\mathbb{E}[\cdot]$ is an expected value over a possible randomization of $F$ applied by the adversary. Thus, the planner chooses a strictly positive $m$ in any equilibrium. However, for $m>0$, $\mathcal{A}$ is strictly decreasing in $\EE_F[\gamma]]$. It is thus optimal for the adversary to choose $\EE_F[\gamma]]$ as large as possible, $\EE_F[\gamma]]=b$. Thus, the adversary must play $F_b$ in any equilibrium. Consequently, the planner must play his optimal response to $F_b$ in any equilibrium which is $m^*(b)$. We have thus derived the unique equilibrium. 
\end{proof}

\begin{proof}[Proof of Proposition \ref{RCG}]
Arguing as in the proof of Lemma \ref{ACG}, we can write
\begin{align*}
\mathcal{R}(m,F)=&\left((\mu-r)m T-\frac12 m^2 \sigma^2 T \EE_F[\gamma]\right)
 - \left((\mu-r)m^*_F T-\frac12 (m^*_F)^2 \sigma^2 T \EE_F[\gamma]\right).
\end{align*}
Plugging in the definition of $m^*_F$ and simplifying, this becomes 
\begin{equation}\label{RR}
\mathcal{R}(m,F)=(\mu-r)m T-\frac12 m^2 \sigma^2 T \EE_F[\gamma] - 
\frac12 \frac{(\mu-r)^2}{\sigma^2 \EE_F[\gamma]} T.
\end{equation}
Now suppose that the adversary plays some pure or mixed strategy and denote by $\mathbb{E}[\cdot]$ a possible expectation over the distribution of $F$. Arguing exactly like in the proof of Lemma \ref{ACG}, the planner's best response to any strategy of the adversary is to choose
\[
m^*(\mathbb{E}[\EE_F[\gamma]])=\frac{\mu-r}{\sigma^2 \mathbb{E}[\EE_F[\gamma]]} >0.
\]
In particular, the planner plays a pure strategy in any equilibrium. Now consider the adversary's problem of minimizing $\mathcal{R}(m,F)$ for some fixed $m$. Since $\mathcal{R}(m,F)$ is strictly concave in $\EE_F[\gamma]$, the minimum must be attained at one of the extremes, $\EE_F[\gamma] \in \{a,b\}$. Thus, any (mixed or pure) equilibrium strategy of the adversary can only take values in $\{F_a,F_b\}$. We next analyze how the adversary's choice between $a$ and $b$ depends on $m$. To this end, denote by $\Gamma=g^*(m)$ the risk type for which $m$ is the individually optimal decision as defined in \eqref{gstar}. Consider the condition $\mathcal{R}(m,F_a)>\mathcal{R}(m,F_b)$ which means that $F_b$ is a strict best response of the adversary to a planner who plays $m$. Plugging in $m=m^*(\Gamma)$, we can express $\mathcal{R}(m,F_a)$ using $\Gamma$ as
 \begin{equation}\label{RGamma}
\mathcal{R}(m,F_a) = \frac{(\mu-r)^2 \,T}{\sigma^2}\left(\frac{1}{\Gamma}-\frac{a}{2 \Gamma^2}-\frac{1}{2a}
\right)
\end{equation}
and similarly for $\mathcal{R}(m,F_b)$. Thus, we can write $\mathcal{R}(m,F_a)>\mathcal{R}(m,F_b)$ as
\[
\frac{a}{\Gamma^2}+\frac{1}{a} < \frac{b}{\Gamma^2}+\frac{1}{b}. 
\]
After a few manipulations, this condition turns out to coincide with $\Gamma <\sqrt{a b}$ and thus  $m > m^*(\sqrt{a b})$. The adversary's best response correspondence thus looks as follows: If $m > m^*(\sqrt{a b})$, play $F_b$, i.e., if $m$ is high the adversary makes agents risk averse. If $m < m^*(\sqrt{a b})$, play $F_a$. If $m = m^*(\sqrt{a b})$, the adversary is indifferent between playing $F_a$ and $F_b$. We thus conclude that there cannot be pure equilibria: In any pure equilibrium, the adversary must play either $F_a$ or $F_b$. Suppose the adversary always played $F_a$ in equilibrium. The planner's best response to $F_a$ is $m^*(a)$, $m^*(a)>m^*(\sqrt{a b})$. The adversary's best response to $m^*(a)$ is thus $F_b$ and not $F_a$. Thus, there cannot be a pure equilibrium in which the adversary plays $F_a$. By similar reasoning,  there is no pure equilibrium in which the adversary plays $F_b$. 

We are now ready to pin down the unique mixed equilibrium. Since the support of the adversary's strategy must be $\{F_a,F_b\}$, we know that such a mixed strategy must take the form of playing $F_a$ with some probability $p \in (0,1)$ and $F_b$ otherwise. For such mixing to be optimal, the adversary must be indifferent between playing $F_a$ and $F_b$. We saw that this indifference can only hold if the planner plays the pure strategy $m^*(\sqrt{a b})$. Thus, to achieve an equilibrium, the adversary must mix over the set $\{F_a,F_b\}$ in such a way that $m^*(\sqrt{a b})$ is the planner's best response. This is equivalent to 
\[
\sqrt{a b} = \mathbb{E}[\EE_F[\gamma]] = p \EE_{F_a}[\gamma]+(1-p)
\EE_{F_b}[\gamma] = pa + (1-p)b.
\]
This equation can always be solved for a unique $p$, as the right hand side is continuous and strictly monotonic in $p$, interpolating between $a$ and $b$ with $a <\sqrt{a b} <b $. This proves existence of a unique equilibrium. Solving for $p$  shows that 
\[
p=
\frac{b-\sqrt{ab}}{b-a}=
\frac{\sqrt{b}(\sqrt{b}-\sqrt{a})}{(\sqrt{b}-\sqrt{a})(\sqrt{b}+\sqrt{a})}
=
\frac{\sqrt{b}}{\sqrt{a}+\sqrt{b}}.
\]
It remains to compute the equilibrium value of $\mathcal{R}$. Denote by $\mathbb{E}^*$ the expected value of $F$ chosen according to the adversary's equilibrium strategy. By construction of the mixed equilibrium, we must have
\[
\mathbb{E}^*\left[\mathcal{R}\left(m^*(\sqrt{a b}),F\right)\right]= \mathcal{R}\left(m^*(\sqrt{a b}),F_a\right).
\]
By \eqref{RR}, we can rewrite this into
\begin{align*}
\mathbb{E}^*\left[\mathcal{R}\left(m^*(\sqrt{a b}),F\right)\right]&= 
\frac{(\mu-r)^2 T}{2\sigma^2} \left( 
\frac{2}{\sqrt{ab}}-\frac1a-\frac1b
\right)
= 
- \,\frac{(\mu-r)^2\, T}{2\sigma^2} \left( 
\frac1{\sqrt{a}}-\frac1{\sqrt{b}}
\right)^2.
\end{align*}
This concludes the proof. \end{proof}

\begin{proof}[Proof of Proposition \ref{nRCG}]
As a first step, notice that in any equilibrium the adversary must randomize over more than 
$n$ distributions of risk types. Otherwise, the planner could just implement the optimal strategies for all possible distributions of risk types and achieve $\mathcal{R}=0$ which is his best possible outcome. The adversary can easily do better than this. As a second step, notice that for any given pure strategy of the planner, there are at most $n+1$  risk types which might appear in a best response of the adversary. To see this, suppose that the planner's strategy is some menu $m_1 > \ldots > m_n$.\footnote{Restricting attention to strictly decreasing sequences is without loss of generality. If the planner would choose less than $n$ distinct $m_i$, the number of potential best responses of the adversary is reduced accordingly, arguing in the same way.} To understand the adversary's possible best responses, we can focus on his pure strategies -- even though mixing over these would be required in equilibrium. Since the adversary's goal is to create a bad match between risk types and available strategies, we can focus on degenerate distributions $F_x$ where all mass is concentrated on a single risk type $x$. The risk type $x$ is chosen as unsuitable as possible for the available strategies $m_i$. Inspecting the objective, we see that the candidates for these worst possible locations of $x$ are the interval boundaries $a$ and $b$ and the $n-1$ points $g_i$ at which the corresponding risk type is indifferent between strategies $m_i$ and $m_{i+1}$.  It follows that there are only $n+1$ candidates for the adversary-optimal location of $x$. By Lemma \ref{lemHM2}, the points $g_i$ are determined by the strategies $m_i$ via the harmonic mean condition
\[
\mathcal{H}(g^*(m_i),g^*(m_{i+1}))=g_i.
\]
Combining our two observations, it follows that in any equilibrium in which the planner plays a pure strategy, implementing a menu $m_1 > \ldots > m_n$, the $m_i$ must have the property that the adversary is \textit{indifferent} between the resulting $n+1$ candidates for the risk types he could choose in equilibrium. Otherwise, it cannot be optimal for the adversary to mix over all $n+1$ candidates.  To complete the proof, we need to show that these indifference conditions uniquely pin down the numbers $m_i$ to be $m_i^*$ given in the proposition. 

We begin by verifying that the solution given in the proposition has all the properties we need. As a first step, observe that the sequence $h_i$ is linear and decreasing from $h_0=\sqrt{b}$ to $h_n=\sqrt{a}$. It follows that the sequences $g_i^*$ and $\Gamma_i^*$ are increasing and contained in the interval $[a,b]$, that $g_0^*=a$, $g_n^*=b$ and that $g_{i-1}^*< \Gamma_i^* < g_{i}^*$ for all $i$. We also have the harmonic mean property
\[
\mathcal{H}(\Gamma_i^*, \Gamma_{i+1}^*)=\frac{ab}{h_i\left(\frac{1}{2}h_{i-1}+\frac{1}{2}h_{i+1}\right)}= \frac{ab}{h_i^2}=g_i^*.
\]
Thus, when the planner offers the menu of choices $m_i^*=m^*(\Gamma_i^*)$, agents will sort themselves according to the partition defined by the boundaries $g_i^*$. 

In the next step, we show that, in response to our strategy for the planner, the adversary is indifferent between the strategies $F_{g_i^*}$, $i=0,\ldots,n$,  which put all mass on risk type $g_i^*$. Moreover, as already argued above, the adversary prefers these $n+1$ strategies over all other strategies. To this end, consider the outcome when the adversary plays $F_g$ for some $g\in [g^*_{i-1},g^*_i]$ so that the resulting agents pick strategy $m_i^*$. Arguing like in the derivation of formula \eqref{RGamma}, this leads to the outcome
\begin{equation}\label{RGamma2}
\mathcal{R}(m^*(\Gamma_i^*),F_g)= Z\cdot \left( \frac{1}{\Gamma_i^*} -\frac{g}{2{\Gamma_i^*}^2}-\frac{1}{2g} \right)\;\; \text{ where }\;\; Z=\frac{\mu-r}{\sigma^2}T>0.
\end{equation}
Since $\mathcal{R}(m^*(\Gamma_i^*),F_g)$ is concave in $g$, the adversary who minimizes it can restrict attention to $g\in \{g^*_{i-1},g^*_i\}$. To show that the adversary is indifferent between the strategies $g^*_i$, we thus have to show that 
$\mathcal{R}(m^*(\Gamma_i^*),F_{g^*_i})$ is the same as $\mathcal{R}(m^*(\Gamma_i^*),F_{g^*_{i-1}})$. To this end, observe that 
\[
\mathcal{R}(m^*(\Gamma_i^*),F_{g^*_i}) = \frac{Z}{2ab} \left(2 h_{i-1}h_i -\frac{h_{i-1}^2h_i^2}{2h_i^2}-h_i^2 \right)= - \frac{Z}{2ab} (h_{i-1}-h_i)^2,
\]
that 
\[
\mathcal{R}(m^*(\Gamma_i^*),F_{g^*_{i-1}}) = \frac{Z}{2ab} \left(2 h_{i-1}h_i -\frac{h_{i-1}^2h_i^2}{2h_{i-1}^2}-h_{i-1}^2 \right)= - \frac{Z}{2ab} (h_{i-1}-h_i)^2,
\]
and that $h_{i-1}-h_i=\frac1n(\sqrt{b}-\sqrt{a})$ does not depend on $i$. In particular, the resulting outcome $\mathcal{R}(m^*(\Gamma_i^*),F_{g^*_i})$ does not depend on $i$ and coincides with $\mathcal{R}^*$ given in the proposition. 

We have thus verified that our proposed sequences $g_i^*$ and $\Gamma_i^*$ have the desired indifference properties. To conclude the proof, we need to show uniqueness, i.e., we need to show that there exists at most one sequence with these properties. Our strategy of proof is as follows. We fix a lowest risk type $a$ and a level $\mathcal{R}<0$ for the outcome of the game.\footnote{Recall that by construction $\mathcal{R}$ cannot take positive values.} We then show that there is at most one sequence of numbers $g_0<\Gamma_1<g_1<\Gamma_2<g_2<\ldots$ which can be constructed iteratively from the requirements that $g_0=a$,
\[
\mathcal{R}(m^*(\Gamma_i),F_{g_{i-1}})=\mathcal{R}
\;\;\text{and}\;\;
\mathcal{R}(m^*(\Gamma_i),F_{g_i})=\mathcal{R}.
\]
We argue that the numbers $g_i$ and $\Gamma_i$ are strictly decreasing in $\mathcal{R}$ for fixed $a$. It follows that there can be at most one value of $\mathcal{R}$ that leads to $g_n=b$. This is the desired uniqueness. To conclude the proof, we thus prove the following two claims: 

\noindent\textbf{Claim 1:} Fix $a>0$ and $\mathcal{R} < 0$. Then, for every $\Gamma>a$ there is a unique $g^*>\Gamma$ such that 
$
\mathcal{R}(m^*(\Gamma),F_{g^*})=\mathcal{R}.
$
 Moreover, $g^*$ is strictly increasing in $\Gamma$ and strictly decreasing in $\mathcal{R}$.

\noindent\textbf{Claim 2:} Fix $a>0$ and $\mathcal{R}<0$. Then, for every $g\geq a$ that satisfies $\mathcal{R}>-Z/(2g)$ there is a unique $\Gamma^*>g$ such that 
$
\mathcal{R}(m^*(\Gamma^*),F_{g})=\mathcal{R}.
$
 If $\mathcal{R}>-Z/(2g)$ is violated, no such $\Gamma^*>g$ exists. When it exists, $\Gamma^*$ is strictly increasing in $g$ and strictly decreasing in $\mathcal{R}$.

The first claim shows that we can uniquely recover $g_i$ from $\Gamma_i$ while the second claim shows that we can uniquely recover $\Gamma_i$ from $g_{i-1}$ provided that it exists. Thus, $g_0=a$ and $\mathcal{R}$ pin down the entire sequences of $g_i$ and $\Gamma_i$. Moreover, the monotonicity properties imply that a decrease in $\mathcal{R}$ shifts the entire sequence upwards. There is thus at most one level of $\mathcal{R}$ which leads to $g_n=b$. This proves uniqueness.\footnote{The existence result in the second claim is conditional, i.e., there only exists a suitable $\Gamma$ if $\mathcal{R}$ is not too negative compared to the level of $g$. This is not a problem for our proof as we are merely interested in uniqueness at this point, having settled existence in a constructive way. Intuitively, existence of $\Gamma$ means that it is possible to find a strategy $m^*(\Gamma)$ which is so risk averse that it causes a utility loss of $\mathcal{R}$ for an agent of type $g$. This can only work if $\mathcal{R}$ is less severe than the utility loss from an infinitely risk averse strategy which, due to \eqref{RGamma2}, is given by 
$
\lim_{\Gamma \rightarrow \infty} \mathcal{R}(m^*(\Gamma),F_{g}) =-Z/(2g).
$
}

To prove \textbf{Claim 1}, we define $S=-2\mathcal{R}/Z$ and use \eqref{RGamma2} to write $\mathcal{R}(m^*(\Gamma),F_{g})=\mathcal{R}$ as
\begin{equation}\label{gGamma1}
\frac{1}{g}=S+\frac{2}{\Gamma}-\frac{g}{\Gamma^2}.
\end{equation}
The ideas of this proof are visualized in the upper panel of Figure \ref{FIG:claims}. 
Equation  \eqref{gGamma1} describes intersections between the function $1/g$ on the left hand side and a decreasing linear function of $g$ on the right hand side. Evaluated at $g=\Gamma$, the linear function takes the value $S+1/\Gamma$ which is greater than the value of $1/\Gamma$ on the left hand side since $S$ is positive. By the convexity and non-negativity of $1/g$ it follows that \eqref{gGamma1} has a unique solution $g^*$ which satisfies $g^*>\Gamma$. It remains to verify the monotonicity properties. When we decrease $\mathcal{R}$, we increase $S$, thus shifting the linear function on the right hand side of \eqref{gGamma1} upwards. This moves the intersection to the right, increasing  $g^*$. Thus, $g^*$ is decreasing in $\mathcal{R}$. Finally, increasing $\Gamma$ increases the $\Gamma$-dependent term $2/\Gamma-g/\Gamma^2$ on the right hand side in the relevant range $g>\Gamma$, thus again moving the intersection $g^*$ to the right. To see this, note that the derivative of $2/\Gamma-g/\Gamma^2$ with respect to $\Gamma$ can be written as $2(g/\Gamma-1)/\Gamma^2$ which is positive for $g>\Gamma$.

\begin{figure}[!h]
	\begin{subfigure}{0.8\textwidth}
		\hspace{2.7cm}
		\includegraphics[width=0.9\linewidth]{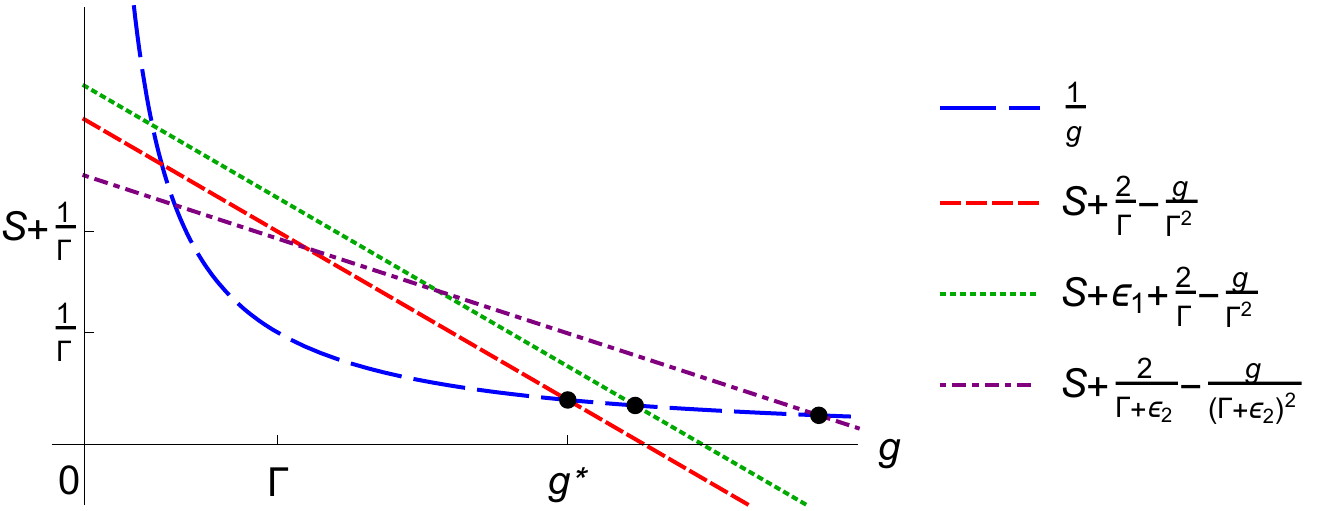}
		\caption*{\hspace{1.5cm} \begin{minipage}[t]{1.00\textwidth}  \footnotesize (a) Proof of Claim 1. The blue and red curves are the left hand side and right hand side of \eqref{gGamma1}. Their intersection in $g^*$ to the right of $\Gamma$ is the existence argument. Moving from the red to the green curve visualizes the comparative statics in $\mathcal{R}$. Moving from the red to the purple curve visualizes the comparative statics in $\Gamma$.\end{minipage}}

	\end{subfigure}
	\begin{subfigure}{.8\textwidth}
	\hspace{2.7cm}
		\includegraphics[width=0.9\linewidth]{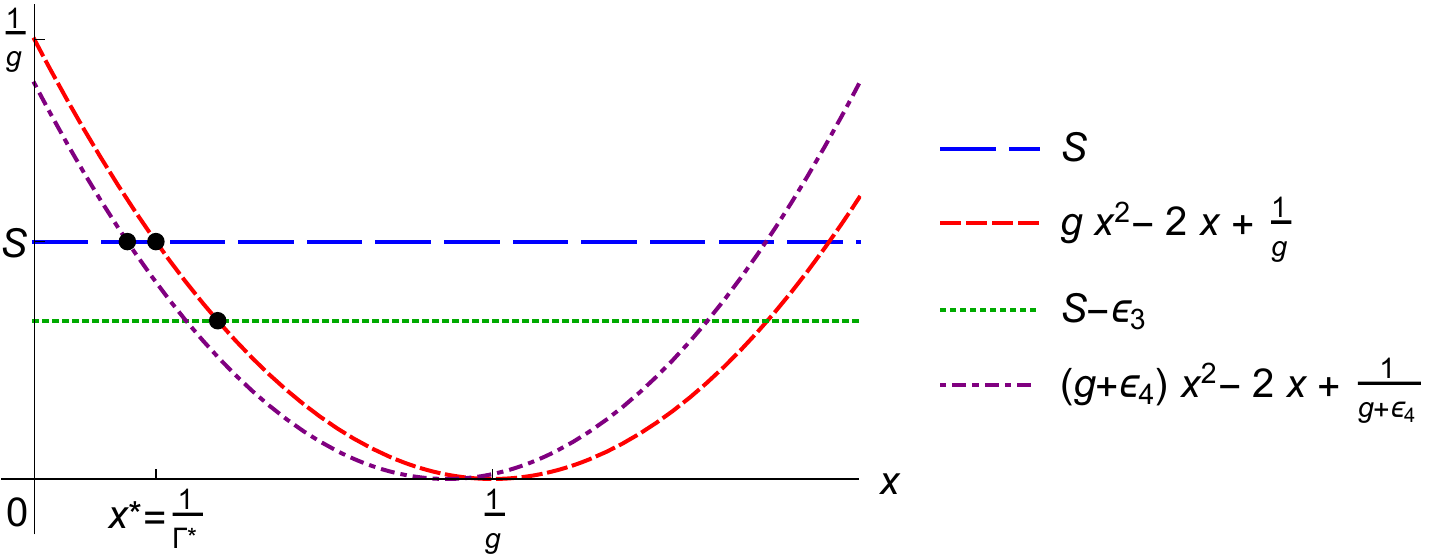}  
\caption*{\hspace{1.5cm} \begin{minipage}[t]{1.00\textwidth}  \footnotesize (b) Proof of Claim 2. The blue and red curves are the left hand side and right hand side of \eqref{gGamma2}. Their intersection in $x^*$ in the interval $(0,1/g)$ is the existence argument. Moving from the blue to the green curve visualizes the comparative statics in $\mathcal{R}$. Moving from the red to the purple curve visualizes the comparative statics in $g$.\end{minipage}}
	\end{subfigure}
	\caption{Claims 1 and 2.}
	\label{FIG:claims}
\end{figure}

To prove \textbf{Claim 2}, we write again $S=-2\mathcal{R}/Z$ and note that our constraint $\mathcal{R}>-Z/(2g)$ can be written as $S<1/g$. Next we change variables, writing  $x=1/\Gamma$. We use \eqref{RGamma2} to write $\mathcal{R}(m^*(\Gamma),F_{g})=\mathcal{R}$ as
\begin{equation}\label{gGamma2}
S=gx^2 -2x +\frac{1}{g}.
\end{equation}
We visualize the ideas of this proof in the lower panel of Figure \ref{FIG:claims}. 
On the right hand side, we have a quadratic polynomial in $x$ which is minimized at $x=1/g$ taking the value $0$. As $S$ is positive, there thus exists a unique $x^*<1/g$ which solves \eqref{gGamma2}. Existence of $x^*$ only translates into existence of a meaningful $\Gamma^*=1/x^*$ when $x^*>0$. To see that this holds, note that for $x=0$ the right hand side of \eqref{gGamma2} takes the value $1/g$ and that $1/g>S$ holds by our constraint on $\mathcal{R}$. Thus, the (left) intersection $x^*$ between the quadratic right hand side and the constant left hand side $S$ satisfies $0<x^*<1/g$. This implies $\Gamma^*=1/x^*>g$. To conclude the proof, we need to argue that $x^*$ increases in $\mathcal{R}$ and decreases in $g$. For $\Gamma^*=1/x^*$, this then implies the opposite monotonicity behavior. When $\mathcal{R}$ increases, $S$ decreases. This moves the two intersections between the left and right hand sides of \eqref{gGamma2} closer together, thus increasing the lower intersecting point $x^*$.  Accordingly, $\Gamma^*$ decreases in $\mathcal{R}$. Finally consider an increase in $g$. This leaves the left hand side of  \eqref{gGamma2} unaffected while the derivative with respect to $g$ of the right hand side is $x^2-1/g^2$ which is negative in the relevant range of $0<x<1/g$. Thus, increasing $g$ decreases the right hand side of  \eqref{gGamma2} around the intersection, moving $x^*$ to the left. Thus, $\Gamma^*=1/x^*$ increases in $g$.
\end{proof}

\begin{proof}[Proof of Corollary \ref{corriab}]
Note that we can write $r_i^*(a,b)$ in terms of $k=b/a$ as follows:
\[
r_i^*=\frac{\left(\frac{i}{n}\sqrt{\frac1k}+\frac{n-i}{n}\right)^{-2}-1}{k-1}.
\] 
We can thus consider $r_i^*$ as a function of $k$. This shows (i). For (ii) and (iii), we compute the limit $k\rightarrow \infty$ of $r_i^*(k)$ by applying L'Hospital's rule and simplifying,
\begin{equation}\label{lhospital}
\lim_{k\rightarrow \infty} r_i^*(k)=  \lim_{k\rightarrow \infty} \; \frac{i}{n} \;\frac{1}{\left(\frac{i}{n}+\frac{n-i}{n}\,\sqrt{k}\right)^3}=0.
\end{equation}
For (iv), we replace the limit in \eqref{lhospital} by a limit $k\downarrow 1$ to obtain the limiting value of $\frac{i}{n}$.
\end{proof}

\singlespacing


\begin{thebibliography}{}

\bibitem[Alsabah et~al., 2020]{cap1}
Alsabah, H., Capponi, A., Ruiz~Lacedelli, O., and Stern, M. (2020).
\newblock Robo-advising: Learning investors' risk preferences via portfolio
  choices.
\newblock {\em Journal of Financial Econometrics}, pages 1--24.

\bibitem[Alserda et~al., 2019]{alserda2019individual}
Alserda, G.~A., Dellaert, B.~G., Swinkels, L., and van~der Lecq, F.~S. (2019).
\newblock Individual pension risk preference elicitation and collective asset
  allocation with heterogeneity.
\newblock {\em Journal of Banking \& Finance}, 101:206--225.

\bibitem[Armbruster and Delage, 2015]{armbruster2015decision}
Armbruster, B. and Delage, E. (2015).
\newblock Decision making under uncertainty when preference information is
  incomplete.
\newblock {\em Management Science}, 61(1):111--128.

\bibitem[Balter et~al., 2021]{Balteretal}
Balter, A.~G., Mahayni, A., and Schweizer, N. (2021).
\newblock Time-inconsistency of optimal investment under smooth ambiguity.
\newblock {\em European Journal of Operational Research}, 293(2):643--657.

\bibitem[Bell, 1982]{bell1982regret}
Bell, D.~E. (1982).
\newblock Regret in decision making under uncertainty.
\newblock {\em Operations Research}, 30(5):961--981.

\bibitem[Ben-Tal et~al., 2009]{ben2009robust}
Ben-Tal, A., El~Ghaoui, L., and Nemirovski, A. (2009).
\newblock {\em Robust Optimization}.
\newblock Princeton University Press.

\bibitem[Branger et~al., 2019]{brangeretal}
Branger, N., Chen, A., Mahayni, A., and Nguyen, T. (2019).
\newblock Optimal collective investment.
\newblock {\em Working Paper, University of M\"unster}.

\bibitem[Calvet et~al., 2021]{calvet2021cross}
Calvet, L.~E., Campbell, J.~Y., Gomes, F., and Sodini, P. (2021).
\newblock The cross-section of household preferences.
\newblock {\em NBER Working Paper}.

\bibitem[Capponi et~al., 2021]{cap2}
Capponi, A., Olafsson, S., and Zariphopoulou, T. (2021).
\newblock Personalized robo-advising: Enhancing investment through client
  interaction.
\newblock {\em Management Science, forthcoming}.

\bibitem[Chen et~al., 2020]{Chenetal}
Chen, A., Nguyen, T., and Rach, M. (2020).
\newblock Optimal collective investment: The impact of sharing rules,
  management fees and guarantees.
\newblock {\em SSRN Preprint 3249094}.

\bibitem[Cronqvist and Thaler, 2004]{cronqvist2004design}
Cronqvist, H. and Thaler, R.~H. (2004).
\newblock Design choices in privatized social-security systems: Learning from
  the swedish experience.
\newblock {\em American Economic Review}, 94(2):424--428.

\bibitem[D'Acunto et~al., 2019]{d2019promises}
D'Acunto, F., Prabhala, N., and Rossi, A.~G. (2019).
\newblock The promises and pitfalls of robo-advising.
\newblock {\em Review of Financial Studies}, 32(5):1983--2020.

\bibitem[D'Acunto and Rossi, 2021]{robosurvey}
D'Acunto, F. and Rossi, A.~G. (2021).
\newblock Robo-advising.
\newblock In {\em The Palgrave Handbook of Technological Finance}. Palgrave.

\bibitem[Dahlquist et~al., 2018]{dahlquist2018asset}
Dahlquist, M., Setty, O., and Vestman, R. (2018).
\newblock On the asset allocation of a default pension fund.
\newblock {\em Journal of Finance}, 73(4):1893--1936.

\bibitem[Desmettre and Steffensen, 2021]{DesmettreSteffensen}
Desmettre, S. and Steffensen, M. (2021).
\newblock Optimal investment with uncertain risk aversion.
\newblock {\em SSRN Preprint 3805069}.

\bibitem[F{\"o}llmer and Schied, 2016]{FollmerSchied}
F{\"o}llmer, H. and Schied, A. (2016).
\newblock {\em Stochastic Finance}.
\newblock de Gruyter.

\bibitem[Hammond, 1992]{Hammond1992}
Hammond, P.~J. (1992).
\newblock Harsanyi's utilitarian theorem: A simpler proof and some ethical
  connotations.
\newblock In Selten, R., editor, {\em Rational Interaction: Essays in Honor of
  John C. Harsanyi}, pages 305--319. Springer.

\bibitem[Hansen and Sargent, 2008]{hansen2008robustness}
Hansen, L.~P. and Sargent, T.~J. (2008).
\newblock {\em Robustness}.
\newblock Princeton University Press.

\bibitem[Jensen and Nielsen, 2016]{jensen2016suboptimal}
Jensen, B.~A. and Nielsen, J.~A. (2016).
\newblock How suboptimal are linear sharing rules?
\newblock {\em Annals of Finance}, 12(2):221--243.

\bibitem[Klibanoff et~al., 2005]{klibanoff2005smooth}
Klibanoff, P., Marinacci, M., and Mukerji, S. (2005).
\newblock A smooth model of decision making under ambiguity.
\newblock {\em Econometrica}, 73(6):1849--1892.

\bibitem[Kryger and Steffensen, 2010]{kryger2010some}
Kryger, E.~M. and Steffensen, M. (2010).
\newblock Some solvable portfolio problems with quadratic and collective
  objectives.
\newblock {\em University of Copenhagen, PhD thesis}.

\bibitem[Loomes and Sugden, 1982]{loomes1982regret}
Loomes, G. and Sugden, R. (1982).
\newblock Regret theory: An alternative theory of rational choice under
  uncertainty.
\newblock {\em Economic Journal}, 92(368):805--824.

\bibitem[Merton, 1971]{merton1971optimum}
Merton, R.~C. (1971).
\newblock Optimum consumption and portfolio rules in a continuous-time model.
\newblock {\em Journal of Economic Theory}, 3(4):373--413.

\bibitem[Nozick, 1974]{nozick1974anarchy}
Nozick, R. (1974).
\newblock {\em Anarchy, State, and Utopia}.
\newblock Basic Books.

\bibitem[Rawls, 1971]{Rawls}
Rawls, J. (1971).
\newblock {\em A Theory of Justice}.
\newblock Harvard University Press.

\bibitem[Rogers, 2013]{rogers2013optimal}
Rogers, L. C.~G. (2013).
\newblock {\em Optimal Investment}.
\newblock Springer.

\bibitem[Schumacher, 2021]{Schumacher}
Schumacher, J.~M. (2021).
\newblock Asymptotics of the assumed interest rate under optimality, fairness,
  and saturation.
\newblock {\em Working Paper, University of Amsterdam}.

\bibitem[Simon and Zame, 1990]{simon1990discontinuous}
Simon, L.~K. and Zame, W.~R. (1990).
\newblock Discontinuous games and endogenous sharing rules.
\newblock {\em Econometrica}, 58:861--872.

\bibitem[Strotz, 1955]{strotz1955myopia}
Strotz, R.~H. (1955).
\newblock Myopia and inconsistency in dynamic utility maximization.
\newblock {\em Review of Economic Studies}, 23(3):165--180.

\bibitem[Vickrey, 1945]{vickrey1945measuring}
Vickrey, W. (1945).
\newblock Measuring marginal utility by reactions to risk.
\newblock {\em Econometrica}, 13(4):319--333.

\bibitem[von Neumann and Morgenstern, 1944]{vNM}
von Neumann, J. and Morgenstern, O. (1944).
\newblock {\em Theory of Games and Economic Behavior}.
\newblock Princeton University Press.

\end{thebibliography}
\end{document}